\documentclass[12pt]{article}

\usepackage{amsmath,amssymb,amsthm,url,caption,subcaption,hyperref}
\usepackage[pdftex]{graphicx}
\DeclareGraphicsExtensions{.pdf, .jpg, .tif}

\renewcommand{\epsilon}{\varepsilon}
\renewcommand{\emptyset}{\varnothing}

\newcommand{\grad}{\nabla}

\usepackage{color}

\newcommand{\less}{\ \backslash\ } 

\newcommand{\st}{\,\mid\,}

\newcommand{\set}[1]{\{#1\}}

\theoremstyle{plain}

\newtheorem{theorem}{Theorem}

\newtheorem{lemma}[theorem]{Lemma}

\theoremstyle{definition}

\theoremstyle{remark}

\newtheorem{remark}[theorem]{Remark}

\numberwithin{equation}{section}
\numberwithin{theorem}{section}
 
\newcommand{\plus}{(+)}
\newcommand{\minus}{(-)}

\newcommand{\erho}{\rho}

\newcommand{\xqed}[1]{%
  \leavevmode\unskip\penalty9999 \hbox{}\nobreak\hfill
  \quad\hbox{\ensuremath{#1}}}

\usepackage{caption,subcaption}

\newcommand{\D}{\mathcal D} 
\newcommand{\R}{\mathbb R}

\begin{document}

\title{Stationary Black Hole Metrics and Inverse Problems in Two Space Dimensions }
\author{
Gregory Eskin\footnote{Email: \url{eskin@math.ucla.edu}}\ \
and 
Michael Hall\footnote{Email: \url{michaelhall@math.ucla.edu}}
\\
{\small Department of Mathematics, UCLA, Los Angeles, CA 90095-1555}
}

\maketitle


\begin{abstract}
We study the wave equation for a stationary Lorentzian metric in the case of two space dimensions. Assuming that the metric has a singularity of the appropriate form, surrounded by an ergosphere which is a smooth Jordan curve, we prove the existence of a black hole with the boundary  (called  the event  horizon)
that  is piece-wise smooth,  generally  having  corners.  
We consider  a physical  model  of acoustic black hole  whose event  horizon has corners.  In  the end  of the paper  we consider  the determination  of a black 
hole by  the boundary  measurements.
\end{abstract}


\section{Introduction}

Consider the wave equation associated to a stationary metric on $\mathbb{R}^{1+2} \cong \mathbb{R}^1_{x_0} \times \mathbb{R}^2_{(x_1,x_2)}$,
\begin{align}\label{eqn1.1}
\sum_{i,j=0}^2 \frac{1}{\sqrt{g(x)}}\frac{\partial}{\partial x_i} \left( \sqrt{g(x)}g^{ij}(x) \frac{\partial u(x_0,x)}{\partial x_j} \right) = 0, \quad (x_0,x) \in \mathbb{R}^{1+2}. 
\end{align}
Here, $[g^{ij}(x)]_{i,j=0}^2$ is the inverse of $[g_{ij}(x)]_{i,j=0}^2$, where $g_{ij}(x) \in C^\infty(\mathbb{R}^{1+2};\mathbb{R})$ defines a pseudo-Riemannian metric with signature $(+1,-1,-1)$ depending only on $x$, with $g_{ij}(x) = g_{ji}(x)$, and $g(x) = \det [g_{ij}(x)]_{i,j=0}^2$. 


For some choices of $g^{jk}(x)$, equation (\ref{eqn1.1}) has a black hole, i.e.\ a region which disturbances may not propagate out of. These are often called \emph{analogue} or \emph{artificial} black holes, since the metric is in general not a solution of the Einstein equations of general relativity \cite{Wald2010},\cite{FrolovNovikov1998}  (A  precise  definition  of black and white holes  is given below).

Two of the most common examples arising from physical models are \emph{optical black holes} (see \cite{Gordon},\cite{LP},\cite{Philbin2008},\cite{Belgiorno2011dielectric}) and \emph{acoustic black holes} (see \cite{Unruh},\cite{VisserAcoustic}). In optics, equation (\ref{eqn1.1}) is a model for the propagation of light through an inhomogeneous moving medium, while in acoustics, it models the propagation of acoustic waves in a moving fluid. Physicists are interested in physical systems which may contain analogue black holes, as they may be suitable for experimental study, while providing some insight into phenomena of general relativity. A number of other models have been studied, including surface waves, relativistic acoustic waves, Bose-Einstein condensates and others \cite{RousseauxLeonhardt2010},\cite{SchutzholdUnruh2002}, \cite{VisserMolinaParis2010},\cite{FFLKT}. See \cite{Visser2012},\cite{BarceloLiberatiVisser},\cite{NovelloVisserVolovik} for surveys and many references.

We define an \emph{event horizon} for (\ref{eqn1.1}) to be a Jordan curve $S_0 \subseteq \mathbb{R}^2$ such that $\mathbb{R} \times S_0$ is piecewise characteristic and forward null-geodesics either can not pass from the interior to the exterior of $S_0$, or vice versa. See section \ref{null}. In the former case we will say that the region enclosed by $\mathbb{R} \times S_0$ is a \emph{black hole}, and in the latter case we call it a \emph{white hole}.

Let $O = (0,0)$ be a singularity of the metric and assume that $g^{jk}$ behaves near $O$ as in \cite{EskinNonstationary}: When $|x| < \epsilon$, assume that
\begin{align}\label{eqn1.2}
g^{jk}(x) = g_1^{jk}(x) + g_2^{jk}(x),
\end{align}
where $g_1^{jk}$ is similar to an acoustic metric (see also Section \ref{Acoustic}):
\begin{align}\label{eqn1.3}
g_1^{00} = 0,\quad  g_1^{j0} = g_1^{0j} = v^j, j = 1,2,\quad g_1^{jk} = v^jv^k,\ j,k=1,2,
\end{align}
where in polar coordinates $x_1 = r\cos\theta$, $x_2 = r\sin\theta$, $\hat r = (\frac{x_1}{r},\frac{x_2}{r})$, $\hat \theta = (-\frac{x_2}{r},\frac{x_1}{r})$, we have 
\begin{align}\label{eqn1.4}
v = (v^1,v^2) = \frac{b_1}{r}\hat r + \frac{b_2}{r} \hat \theta,
\end{align}
where $b_j = b_j(\theta)$, $j = 1,2$ are smooth with $b_1(\theta) \neq 0$. Assume also that $g_2^{jk}$ is smooth in $(r,\theta)$, with $g_2^{00} \geq C > 0$, $g_2^{j0} = g_2^{0j} = \mathcal{O}(r)$, $1 \leq j \leq 2$, and $[g_2^{jk}]_{j,k=1}^2$ a negative definite matrix when $|x| < \epsilon$:
\begin{align*}
([g_2^{jk}]_{j,k=1}^2 \alpha,\alpha) \leq -C_0 |\alpha|^2, \quad \alpha \in \mathbb{R}^2. 
\end{align*}

Let $\Delta(x) = g^{11}(x)g^{22}(x) - (g^{12}(x))^2$.
Define the \emph{ergoregion} to be the set $\Omega \subseteq \mathbb{R}^2$ where
  $g_{00}(x)<0$.  By the Cramer rule  $g_{00}(x)=\Delta(x)g(x)$  with  $g(x)$  as in  (\ref{eqn1.1}).  Thus
  $\Delta(x) < 0$ is equivalent to $g_{00}(x) < 0$. Assume the boundary $\partial \Omega = \set{\Delta(x) = 0}$, 
  called the \emph{ergosphere}, is a Jordan curve encircling $O$ that is smooth in the sense that the gradient of $\Delta(x)$ is nonzero on $\partial \Omega$. 

In  \cite{EskinInvHyp}, it was shown that if the ergosphere is a smooth characteristic surface or non-characteristic surface which contains a trapped surface, then it contains a black hole or a white hole. See also \cite{EskinNonstationary} and \cite{HallThesis}. In this paper we prove the following much more general result:

\begin{theorem}\label{maintheorem}
Let $g$ be any metric such that the ergosphere $\{\Delta(x) = 0\}$ for equation (\ref{eqn1.1}) is a Jordan curve which is smooth in the sense that 
$\frac{\partial \Delta(x)}{\partial x}\neq 0$ when $\Delta(x)=0$,
and the ergoregion 
$\Omega = \{\Delta(x) < 0\}$ contains a singularity $O$, which satisfies (\ref{eqn1.2})-(\ref{eqn1.4}).
Then there exists a black hole in $\mathbb{R} \times \overline{\Omega}$ if $b_1(\theta) < 0$, and there exists a white hole if $b_1(\theta) > 0$. Moreover, the event horizon may have corner points, while it is continuously differentiable outside these corner points.
\end{theorem}

The plan of the paper is as follows. In Section \ref{null}, we discuss the general behavior of null geodesics for metrics satisfying the hypotheses of Theorem \ref{maintheorem}. In Section \ref{proof}, we prove the existence of a black or white hole and show that the event horizon is $C^1$, except at corner points. 
In Section \ref{Acoustic}, we study acoustic black holes and demonstrate that the event horizon may have corners. 
 In Section 5 we study  the determination  of black holes  horizon  by the boundary measurements  on $\R\times\D$  where  $\D\subset \R^n$  is  
 a bounded  domain  containing  
an   ergoregion  $\Omega$.  Assuming  that  the conditions of Theorem 1.1 are satisfied   and no point of  $\partial\Omega$   is characteristic 
 we prove  that the boundary measurements determine  the black hole's event horizon  inside $\Omega$  up to a change of variables.

\section{Null Geodesics}
\label{null}

\subsection{Zero-energy null geodesics in the ergoregion}

Consider bicharacteristics for the wave equation (\ref{eqn1.1}),
\begin{align*}
\frac{dx_p}{ds} = 2 \sum_{k=0}^2 g^{pk}(x(s))\xi_k(s), \quad \frac{d \xi_p}{ds} = -\sum_{j,k=0}^2 g^{jk}_{x_p}(x(s))\xi_j(s)\xi_k(s), \quad 0 \leq p \leq 2.
\end{align*}

Since the metric is stationary we have that $\xi_0(s)$ is constant. Consider null-bicharacteristics with $\xi_0(s) = 0$. We shall call null-bicharacteristics with $\xi_0(s) = 0$ `zero-energy' null-bicharacteristics. Their projections onto $(x_1,x_2)$ will be called zero-energy null-geodesics. For all $s$, $x = x(s)$ and $(\xi_1,\xi_2) = \xi = \xi(s)$ must satisfy
\begin{align}
\label{charnull}
\sum_{j,k =  1}^2 g^{jk}(x)\xi_j\xi_k = 0, \quad (\xi_1,\xi_2) \neq (0,0),\ \ \ x\in \Omega.
\end{align}
For each $x \in \Omega$ there are two linearly independent solutions $\xi^\pm = (\xi_1^\pm,\xi_2^\pm)$ of (\ref{charnull}). It was shown in \cite{EskinInvHyp}, for $|x| > \epsilon$, and in \cite{EskinNonstationary}, for $|x| < \epsilon$, that there exists a pair of continuous vector fields $f^\pm(x) = (f_1^\pm(x),f_2^\pm(x))$ on $\overline \Omega \less O$, satisfying
\begin{equation}\label{vfs}
\begin{gathered} 
0 \neq f^+(x) = f^-(x), \ x \in \partial \Omega, \\
f^+(x),f^-(x) \text{ linearly independent},\ x \in \Omega \less O,\\
f_1^\pm(x) \xi_1^\pm + f_2^\pm(x) \xi_2^\pm = 0,\ (\xi_1^\pm,\xi_2^\pm)\text{ solving (\ref{charnull})}.  
\end{gathered}
\end{equation}
The choice of sign is arbitrary, but the pair $f^\pm(x)$ is otherwise well-defined up to rescalings which respect (\ref{vfs}). If we parameterize zero-energy null-bicharacteristics $(x^\pm(x_0),\xi^\pm(x_0))$ by $x_0$, then we have
\begin{align}\label{eqn2.3}
\frac{dx_j^\pm}{dx_0} = \frac{g^{j1}(x(x_0))\xi_1^\pm(x_0) + g^{j2}(x(x_0))\xi_2^\pm(x_0)}{ g^{01}(x(x_0))\xi_1^\pm(x_0) + g^{02}(x(x_0))\xi_2^\pm(x_0)}, \quad j = 1,2.
\end{align}
Since $f_1^\pm(x(x_0)) \xi_1^\pm(x_0) + f_2^\pm(x(x_0)) \xi_2^\pm(x_0) = 0$ we have that $\xi_1^\pm(x_0) = f_2^\pm(x(x_0))$, $\xi_2^\pm(x_0) = -f_1^\pm(x(x_0))$ up to a nonzero factor. Substituting into (\ref{eqn2.3}), we get
\begin{align}\label{eqn2.4}
\frac{dx_j^\pm}{dx_0} = \frac{g^{j1}f_2^\pm(x) - g^{j2}f_1^\pm(x)}{g^{10}f_2^\pm(x) - g^{20}f_1^\pm(x)}, \quad j = 1,2.
\end{align}
In other words, \emph{the zero-energy null-geodesics in $\Omega \less O$, are the solutions $x = x^+(x_0)$, $x = x^-(x_0)$ of an autonomous system of differential equations}. We shall call the two families of solution curves or trajectories for (\ref{eqn2.4}) the $\plus,\minus$ families, respectively.

Note that 
\begin{align}\label{eqn2.5}
 \frac{dx_2^\pm}{dx_1^\pm} = \frac{g^{21}f_2^\pm - g^{22}f_1^\pm}{g^{11}f_2^\pm - g^{12}f_1^\pm} = \frac{f_2^\pm}{f_1^\pm}
\end{align}
since $g^{21}f_2^\pm f_1^\pm - g^{22}(f_1^\pm)^2 = g^{11}(f_2^\pm)^2 - g^{21}f_1^\pm f_2^\pm$. 
Since the rank of $[g^{jk}(x)]_{j,k=1}^2$ is equal to 1 in $\partial \Omega$ we get $\frac{dx_j^\pm}{dx_0} = 0$, $j = 1,2$, on $\partial \Omega$, but $\frac{dx_2^\pm}{dx_1^\pm}$ has a limit on $\partial \Omega$. Note also that \cite{EskinInvHyp} 
\begin{align}\label{eqn2.6}
g^{10}f_2^\pm - g^{20}f_1^\pm \neq 0,
\end{align}
As in \cite{EskinInvHyp}, we have $f^\pm(x)\cdot \nabla G^\pm(x) = 0$, where $G^\pm(x) = c^\pm$ are characteristic curves. From (\ref{eqn2.4}), (\ref{eqn2.5}) it follows that this is also true when $f^\pm(x)$ is replaced by the right hand sides of (\ref{eqn2.4}). 

\subsection{Coordinates near $\partial\Omega$.}

Introduce coordinates $(\erho,\theta)$ near $\partial \Omega$, where $\erho = -\Delta(x) \geq 0$ in $\Omega$,
$\theta\in [0,2\pi]$  is a parameter on $\partial\Omega$. 
One can extend such coordinates to the whole domain $\Omega \less O$ but 
we shall only use them when $0 \leq \rho \leq \erho_0$ for some small $\erho_0 > 0$. In $(\erho,\theta)$ coordinates, (\ref{charnull}) is replaced by
\begin{align*}
g^{\erho\erho}\xi_\erho^2 + 2g^{\erho \theta }\xi_\erho\xi_\theta + g^{\theta \theta }\xi_\theta ^2 = 0, \quad (\xi_\erho,\xi_\theta) \neq (0,0).
\end{align*}
Note  that 
$$
(g^{\rho\theta})^2-g^{\rho\rho}g^{\theta\theta}=C^2(\rho,\theta)\rho,
$$
where  $C(\rho,\theta)>0$. We  shall denote  $C^2(\rho,\theta)\rho$  by  $\rho_1(\rho,\theta)$.  Thus  $C\sqrt \rho=\sqrt {\rho_1}$.
Either $g^{\erho \erho}$ or $g^{\theta \theta }$ is not zero when $\erho = 0$ since the rank of $[g^{jk}]_{j,k=1}^2$ is 1 on $\partial \Omega$.
 Let $\erho = 0$, $\theta = \theta_0$ be such that $g^{\theta\theta}(0,\theta_0) \neq 0$. Near $(0,\theta_0)$, we write the solutions
\begin{align*}
\xi_\theta ^\pm = \frac{-g^{\erho\theta} \pm \sqrt{\erho_1}}{g^{\theta \theta }} \xi_\erho^\pm. 
\end{align*}
Then in $(\erho,\theta)$ coordinates, (\ref{eqn2.4}) gives  
\begin{equation}\label{eqn2.7}
\begin{aligned}
\frac{d\erho^\pm}{dx_0} 
&= 
\frac{ g^{\erho\erho}\xi_\erho^\pm + g^{\erho\theta}\xi_\theta^\pm }
{ g^{0 \erho}\xi_\erho^\pm + g^{0 \theta}\xi_\theta^\pm } 
= 
\frac{ g^{\erho\erho} + g^{\erho\theta}\frac{-g^{\erho\theta}\pm \sqrt{\erho_1}}
{g^{\theta \theta }} }
{ g^{0 \erho}+ g^{0 \theta} \frac{-g^{\erho\theta}\pm \sqrt{\erho_1}}{g^{\theta \theta} }}
= 
\frac{-\erho_1\pm g^{\erho \theta}\sqrt{\erho_1}}{b(\erho,\theta)\pm g^{0\theta}\sqrt{\erho_1}}
\\
\frac{d \theta^\pm}{dx_0}
&= 
\frac{g^{\erho\theta}\xi_\erho^\pm + g^{\theta \theta}\xi_\theta^\pm}
{g^{0\erho}\xi_\erho^\pm + g^{0 \theta}\xi_\theta^\pm}
= 
\frac{[g^{\erho\theta}+(-g^{\erho\theta}\pm \sqrt{ \erho_1})]}
{g^{0\erho} + g^{0\theta} \frac{-g^{\erho\theta}\pm \sqrt{\erho_1}}{g^{\theta \theta}}}
= 
\frac{ \pm g^{\theta \theta } \sqrt\erho_1}
{b(\erho,\theta) \pm g^{0\theta}\sqrt{\erho_1}},
\end{aligned}
\end{equation}
where $b(\erho,\theta) = g^{0\erho}g^{\theta\theta} - g^{0\theta}g^{\erho\theta} \neq 0$ (see (\ref{eqn2.6})).  

\subsection{Types of boundary points.}
If $g^{\erho\theta}(0,\theta_0) \neq 0$ then 
\begin{align}\label{eqn2.8}
\frac{d\erho^\pm}{d\theta} = \mp \frac{\sqrt{\erho_1}}{g^{\theta\theta}} + \frac{g^{\erho\theta}}{g^{\theta\theta}}
\end{align}
is not zero near $(0,\theta_0)$, i.e.\ the curve $\erho = \erho^\pm(\theta)$ is transverse to the boundary $\erho = 0$ near $(0,\theta_0)$. It follows from (\ref{eqn2.7}) that the trajectories $(\erho^\pm(x_0),\theta^\pm(x_0))$ reach the boundary $\erho = 0$ in finite time. Since 
\begin{equation}\label{eqn2.9}
\frac{d\erho^\pm}{dx_0} = \pm \frac{g^{\erho\theta}(0,\theta)}{b(0,\theta)}\sqrt{\erho_1} + \mathcal{O}(\erho)
\end{equation}
near $(0,\theta_0)$, one family of trajectories approaches the boundary as $x_0$ increases while the other leaves the boundary as $x_0$ increases. 

Make a change of variables $w = \sqrt\erho$.
Denote  $w_1=\sqrt{\rho_1}=Cw$.
 Since $\frac{d\erho}{dx_0} = 2w \frac{dw}{dx_0}$, so we get
\begin{align}\label{eqn2.10}
2 \frac{dw}{dx_0} = \frac{ \pm C g^{\erho\theta}(w^2,\theta) - C^2 w}
{b(w^2,\theta) \pm g^{0\theta}(w^2,\theta)w}, 
\quad 
\frac{d \theta }{dx_0} = \frac{ \pm C g^{\theta \theta}(w^2,\theta)w}
{b(w^2,\theta) \pm g^{0\theta}(w^2,\theta)w},
\end{align}
where $b(0,\theta_0) \neq 0$, $g^{\theta \theta}(0,\theta_0) \neq 0$. 
If $g^{\erho \theta}(0,\theta_0) = 0$ then $(0,\theta_0)$ is a tangential point of $\partial\Omega$. 

If $\frac{\partial  g^{\erho \theta}}{\partial \theta}(0,\theta_0) \neq 0$ then $(0,\theta_0)$ is a non-degenerate critical 
point in $(w,\theta)$ coordinates. It could be a node, saddle, degenerate node, or spiral restricted to the half-space 
$w \geq 0$  (cf.  Fig. 2-4).

Consider now the case when $g^{\theta\theta}(0,\theta_0) \neq 0$, $g^{\erho\theta}(0,\theta_0) = 0$, and $g^{\erho\theta}_\theta(0,\theta_0) = 0$. 
To fix ideas suppose $g^{\theta\theta}(0,\theta_0) < 0$. Then the equation (\ref{eqn2.8}) has the following form in $(w,\theta)$ coordinates:
\begin{align}\label{eqn2.11}
2w \frac{dw^\pm}{d\theta} = \frac{\mp C w+ g^{\erho\theta}(w^2,\theta)}{g^{\theta\theta}(w^2,\theta)}
\end{align}
\begin{lemma}\label{Lemma2.1}
There is a $\plus$ solution of (\ref{eqn2.11}) satisfying $$w_*^+(\theta) = a_1(\theta) + w_1^+(\theta), \quad |w_1^+(\theta)| \leq C|\theta-\theta_0|^2,$$ defined on $(\theta_0,\theta_0+\delta)$, $\delta > 0$ small, where 
$$a_1(\theta) = \int_{\theta_0}^\theta a(\theta')\, d\theta', \quad a(\theta) = -\frac{1}{2g^{\theta\theta}(0,\theta)}.$$

Analogously, there is a $\minus$ solution of (\ref{eqn2.11})  satisfying $$w_*^-(\theta) = -a_1(\theta) + w_1^-(\theta), \quad |w_1^-(\theta)| \leq C|\theta-\theta_0|^2,$$ defined on $(\theta_0-\delta,\theta_0)$, $\delta$ small, with $a_1(\theta)$ as above. 
\end{lemma}

\begin{proof}
We rewrite equation (\ref{eqn2.11}) 
\begin{align*}
\frac{dw^+(\theta)}{d\theta} 
&= a(\theta) + \frac{g^{\erho\theta}(0,\theta)}{2w g^{\theta\theta}(0,\theta)} + \frac{g_1(w^2,\theta)}{w},
\end{align*}
where $|g_1(w^2,\theta)| \leq Cw^2$. Let $g_2(w,\theta) = \frac{g_1(w^2,\theta)}{w}$, 
$g_3(\theta) = \frac{g^{\erho\theta}(0,\theta)}{2g^{\theta\theta}(0,\theta)}$. Then for $w_1^+(\theta)$ we get
\begin{align}\label{eqn2.12}
\frac{dw_1^+}{d\theta} = \frac{g_3(\theta)}{a_1(\theta) + w_1^+(\theta)} + g_2(a_1(\theta) + w_1^+(\theta)), \quad w_1^+(\theta_0) = 0.
\end{align}
Let $B$ be the Banach space with norm $\| h \| = \sup_{\theta_0 \leq \theta \leq \theta_0 + \delta} \frac{|h(\theta)|}{(\theta-\theta_0)^2}$. The integral from $\theta_0$ to $\theta$ of the right hand side of (\ref{eqn2.12}) is a contraction mapping in $B$ if $\delta$ is small. Therefore $w_1^+(\theta)$ exists. 

The proof for $w_1^-(\theta)$ is similar. 
\end{proof}

\begin{remark}
Note that the lemma remains valid when $g^{\erho\theta}_\theta(0,\theta_0) \neq 0$ but is small. \xqed{\lozenge}
\end{remark}

\begin{remark}\label{rmk2.2}
Suppose $\partial \Omega$ contains a characteristic segment $L$. Let $\hat x$ be an interior point of $L$.  Since the boundary $w = 0$ is characteristic 
for all $\theta$ in a neighborhood of $\hat x$ we have $g^{\erho \theta}(0,\theta) = 0$, i.e.\ $g^{\erho \theta}(w^2,\theta) = \mathcal{O}(w^2)$. 
Therefore by (\ref{eqn2.11}),
\begin{align*}
\frac{dw^\pm}{d \theta} = \frac{\mp C}{2g^{\theta \theta}(w^2,\theta)} + \mathcal{O}(w).
\end{align*}
Note that $g^{\theta \theta}(w^2,\theta) \neq 0$ near $\hat x$. Also, 
\begin{align*}
\frac{dw^\pm}{d x_0} = \frac{ -w C^2}{2b(w^2,\theta)} + \mathcal{O}(w^2). 
\end{align*}
Since $b(w^2,\theta) < 0$ we have that $w^\pm(x_0)$ increases when $x_0$ increases. 
Therefore we have two zero-energy null-geodesics on the set $w \geq 0$ that start at $\hat x$. 

In $(\erho,\theta)$ coordinates these two zero-energy null-geodesics are tangent to the boundary $\erho = 0$. The same picture is true for 
any $\theta_1$ close to $\theta_0$. Note that $w = 0$ is an envelope of both the $\plus$ and $\minus$ families near $(0,\theta_0)$. 
Also  note  that  the point  $(0,\theta_0)$  where  $g^{\theta\theta}(0,\theta_0)=0$     is not  characteristic.
\qed
\end{remark}

\section{Existence of a black hole}
\label{proof}

We shall consider the case when $b_1(\theta) < 0$ and show the existence of a black hole. The case when $b_1(\theta) > 0$ may be treated similarly. 

Consider a small circle $\{ |x| = \epsilon \}$ around $O$. Since $b_1 < 0$, an integral curve of either the $\plus$ or $\minus$ family starting at $\{ |x| = \epsilon \}$ goes to $O$ as $x_0$ increases, i.e.\ $\set{|x| < \epsilon}$ is a trapped region; see \cite{EskinInvHyp}. Let $\Omega^+$ be the union of all trajectories of the $\plus$ family in $\Omega \less \{ |x| \leq \epsilon \}$ which end at $\{ |x| = \epsilon \}$, i.e.
\begin{multline}
\Omega^+ = \{ x^+(x_0) \st x_0 \in (\ell,0)\ \mbox{where}\ -\infty\leq \ell<0;\ x^+\text{ solves }(\ref{eqn2.4}) ; \\
x^+(x_0) \in \Omega \text{ for } x_0 \in (\ell,0);\ x^+(0) \in \{ |x| = \epsilon \} ; \text{ and } x^+(\ell) \in \partial \Omega \text{ when }\ell > -\infty \}.
\end{multline}
%

\begin{lemma}\label{twoposs}
Suppose $z_0 \in \partial\Omega^+$ is an interior point of $\Omega$. Let $\gamma_0^+$ be a curve of the $\plus$ family passing through $z_0$, parameterized
 $x = x^+(x_0)$.
Then there are two possibilities:
\begin{enumerate}
\item $\gamma_0^+$ is a characteristic segment with endpoints $\alpha_1,\alpha_2 \in \partial \Omega$, with $\gamma_0^+$ tangent to $\partial \Omega$ at 
$\alpha_1 = \lim_{x_0 \to \infty} x^+(x_0)$.

\item $\gamma_0^+$ is a smooth closed periodic orbit.
\end{enumerate}
In both cases, $\gamma_0^+ \subseteq \partial \Omega^+$.
\end{lemma}

\begin{proof}
Note that $z_0 \not \in \Omega^+$ since $\Omega^+$ is open. First suppose the curve $\gamma_0^+$ has endpoints $\alpha_1,\alpha_2 \in \partial \Omega$ 
with $x^+(x_0)$ directed toward $\alpha_1$ when $x_0$ increases. Since $z_0$ is an interior point of $\Omega$, there is a small neighborhood 
$\mathcal{U}_\epsilon$ of $z_0$ contained in $\Omega$. The curves of the $\plus$ family passing through points of $\mathcal{U}_\epsilon$ 
form a ``strip'' $V_\epsilon$.

Since $z_0 \in \partial \Omega^+$ there exist $z_n,z_n' \in \mathcal{U}_\epsilon$, $z_n \to z_0$, $z_n' \to z_0$ 
such that $z_n \in \Omega^+$, $z_n' \not\in \Omega^+$. Therefore there are $\plus$ trajectories $x_n(x_0)$ in the strip $V_\epsilon$ 
belonging to $\Omega^+$ with $x_n(x_{0n}) = z_n$, and trajectories $x_n'(x_0)$ not belonging to $\Omega^+$ with $x_n'(x_{0n}') = z_n'$. 
If $z^{\scriptscriptstyle(1)}$ is any other interior point of $\gamma_0^+$ then the trajectories  $x_n(x_0)$ and $x_n'(x_0)$ come arbitrarily close to 
$z^{\scriptscriptstyle(1)}$. Therefore $z^{\scriptscriptstyle(1)} \in \partial \Omega^+$. 

We claim $\gamma_0^+$ is tangent to $\partial \Omega$ at $\alpha_1$. Indeed if $\gamma_0^+$ is transversal to $\partial \Omega$ at $\alpha_1$, then all $\plus$ curves in $V_\epsilon$ also intersect $\partial \Omega$ transversally when $\epsilon > 0$ is small. Therefore all $\plus$ curves in $V_\epsilon$ end at $\partial \Omega$, and do not reach $\{|x| = \epsilon \}$. This contradicts the fact that $V_\epsilon$ contains $\plus$ curves belonging to $\Omega^+$. 

To show $\alpha_1 = \lim_{x_0 \to +\infty} x^+(x_0)$, we  use (\ref{eqn2.10}). Since $g^{\erho\theta}(0,\theta_0) = 0$ we have $|g^{\erho\theta}(w^2,\theta)|
 \leq C(w + |\theta - \theta_0|)$. Thus, 
\begin{align*}
\left| \frac{dw}{dx_0} \right| \leq C(w + |\theta - \theta_0|), \quad \left| \frac{d(\theta-\theta_0)}{dx_0} \right| \leq Cw. 
\end{align*}
Therefore
\begin{align*}
\left| \frac{d(w + |\theta-\theta_0|)}{dx_0} \right| \leq C(w+ |\theta - \theta_0|),
\end{align*}
and 
\begin{align*}
|dx_0| \geq \frac{1}{C} \frac{d(w + |\theta-\theta_0|)}{w + |\theta - \theta_0|}.
\end{align*}
Hence $x_0 \to +\infty$ when $w + |\theta-\theta_0| \to 0$. At the point $\alpha_2$ the curve $\gamma_0^+$ 
may be either transversal or tangent to $\partial \Omega$. If it is tangent then, analogously, $\alpha_2 = \lim_{x_0 \to -\infty} x^+(x_0)$.

If $\gamma_0^+$ can be extended indefinitely when $x_0 \to +\infty$ or $-\infty$ without approaching $\partial \Omega$ then the corresponding limit set of $\gamma_0^+$ is a closed orbit $\gamma_1^+$ by the Poincar\' e-Bendixson theorem. Since $\gamma_0^+ \subseteq \partial \Omega^+$, we also have
 $\gamma_1^+ \subseteq  \partial \Omega^+$, and hence $\gamma_0^+ = \gamma_1^+ = \partial \Omega^+$. This concludes the proof of Lemma \ref{twoposs}.
\end{proof}

\subsection{The case of finitely many tangential points}
\label{finite}

\subsubsection{Construction of the event horizon}

Suppose the ergosphere $\partial \Omega$ has a finite number of points $\alpha_1, \ldots ,\alpha_m$ such that the normals to
 $\partial \Omega$ at $\alpha_p$, $1 \leq p \leq m$ are characteristic directions, i.e.\ $\sum_{j,k=1}^2 g^{jk}(\alpha_p)\nu_j(\alpha_p)\nu_k(\alpha_p) = 0$ 
 where $\nu = (\nu_1,\nu_2)$ is the outward normal to $\partial \Omega$.  In other words, the vector fields $f^\pm$ 
 are tangent to $\partial \Omega$ at $x = \alpha_p$, $1 \leq p \leq m$. 

As in Lemma \ref{twoposs}, let $z_0 \in \partial \Omega^+$ be an interior point of $\Omega$, and let $\gamma_0 \subseteq  \partial \Omega^+$ be a characteristic curve passing through $z_0$. Suppose that $\gamma_0$ can be continued indefinitely as $x_0$ decreases and does not approach $\partial \Omega$, and hence $\gamma_0$ is a closed periodic orbit belonging to the $\plus$ family. If the trajectories of the $\minus$ family are directed inside $\gamma_0$ when $x_0$ increases then $\mathbb{R} \times \gamma_0$ is a black hole event horizon, while if they are directed outside it is a white hole event horizon. In the latter case any trajectory of the $\minus$ family ending at $S_\epsilon = \{ |x| = \epsilon  \}$ can not reach $\gamma_0$ as $x_0$ decreases. Therefore by the Poincar\' e-Bendixson theorem there exists a periodic orbit $\gamma_1$ inside the domain bounded by $\gamma_0$ and belonging to the $\minus$ family. Then $\mathbb{R} \times \gamma_1$ is a black hold event horizon. 

Suppose that instead $\gamma_0$ is a characteristic segment connecting points $\beta_{02}, \beta_{01} \in \partial \Omega$, where by Lemma \ref{twoposs} at least one of $\beta_{0j}$ must be a characteristic point. Then $\gamma_0$ divides the domain $\Omega$ into parts $\Omega_1, \Omega_1'$ and we shall assume that $\Omega_1$ contains $\Omega^+$. Then consider $\Omega_1$ instead of $\Omega$. Suppose $z_1 \in \partial \Omega^+$ is in the interior of $\Omega_1$, and let $\gamma_1$ be a characteristic segment with endpoints $\beta_{11}, \beta_{12} \in \partial \Omega_1$, where at least one of $\beta_{11},\beta_{12}$ is a tangential point. 

Since $\gamma_0$ and $\gamma_1$ belong to the $\plus$ family it is impossible that an endpoint of $\gamma_1$ belongs to the interior of $\gamma_0$, and similarly an endpoint of $\gamma_0$ does not coincide with an endpoint of $\gamma_1$ unless both curves are tangential to $\partial \Omega$ at this point. Thus the only possibilities are that $\gamma_0,\gamma_1$ do not intersect, or they share a common tangential point in $\partial \Omega$.
The curve $\gamma_1$ divides $\Omega_1$ into pieces $\Omega_2,\Omega_2'$, where $\Omega_2$ contains $\Omega^+$. Replace $\Omega_1$ by $\Omega_2$. 

If there is a point $z_2 \in \partial \Omega^+$ such that $z_2$ is an interior point of $\Omega_2$, we repeat the previous argument, etc. After finitely many steps we get a domain $\Omega_r$ such that $\overline{ \Omega^+ } =\overline{ \Omega_r }$ and $\Omega_r$ is a domain whose boundary consists of a finite number of characteristic segments $\gamma_0,\gamma_1, \ldots ,\gamma_{r-1}$ of the $\plus$ family and a finite number of segments $\delta_1, \ldots ,\delta_q$ of $\partial \Omega$. Since $\delta_j \subseteq \partial \Omega \cap \overline{ \Omega^+ }$, $1 \leq j \leq q$, the $\plus$ family of solutions must be directed into $\Omega^+$ on $\cup_{j = 1}^q \delta_j$. 

Denote by $\Omega_r^-$ the union of all trajectories of the $\minus$ family in $\Omega_r$ that end on $S_\epsilon$. Note that $\Omega_r^-$ is an open set.

Since the open segments $\delta_j$, $1 \leq j \leq q$, are not characteristic, and trajectories of the $\plus$ family start on $\delta_j$, we have that trajectories of the $\minus$ family end on $\delta_j$, $1 \leq j \leq q$. Note that no trajectory of the $\minus$ family that ends on $S_\epsilon$ can end on $\delta_j$, $1 \leq j \leq q$. This means that $\overline{\Omega_r^-}$ does not touch the interior of $\delta_j$, $1 \leq j \leq q$.

Now  apply  to the  trajectories  of $(-)$  family
in $\Omega_r$  and $\Omega_r^-$
  the same arguments  as for the trajectories of  $(+)$  family  in $\Omega$  and  $\Omega^+$

 After a finite number of steps we get a domain $\Omega_{r+p}$ such that $\overline{ \Omega_{r+p} } = \overline{ \Omega_r^-}$ 
and the boundary of $\Omega_{r+p}$ consists of a finite number of characteristic segments $\gamma_0^-, \ldots ,\gamma_{p-1}^-$
 of the $\minus$ family and some of the characteristic segments $\gamma_0, \ldots ,\gamma_{r-1}$, or parts of them, belonging to the $\plus$ family. 
 Since some of the segments $\gamma_k$ may have been truncated by the above procedure, the boundary of $\Omega_{r+p}$ may not be smooth, 
 as some $\gamma_j^-,\gamma_k$ may intersect at a corner (cf. Section \ref{Acoustic}).

\subsubsection{The domain $\mathbb{R} \times \Omega_{r+p}$ is a black hole}

To show that $\mathbb{R} \times \Omega_{r+p}$ is a black hole we shall show that any point $(\hat x_0, \hat x) \in \mathbb{R} \times \partial \Omega_{r+p}$ is a no-escape point. More precisely, let $K_+(\hat x)$ be the forward light cone at $(\hat x_0,\hat x)$, consisting of all $(\dot x_0,\dot x_1, \dot x_2) \in \mathbb{R}^1 \times \mathbb{R}^2$ such that $\sum_{j,k=0}^2 g_{jk}(\hat x) \dot x_j \dot x_k > 0$, $\dot x_0 > 0$. Denote by $\Pi^\pm_{\nu(\hat x)}$ the half-space $\set{(\alpha_0,\alpha_1,\alpha_2) \st \alpha_1 \nu_1 + \alpha_2 \nu_2 \gtrless 0}$, where $\nu(\hat x) = (\nu_1,\nu_2)$ is the outward normal to $\partial \Omega_{r+p}$ at $x = \hat x$. Then $\hat x$ is a point of no escape if $K_+(\hat x) \subseteq \Pi_{\nu(\hat x)}^-$ for all $x_0 \in \mathbb{R}$, see \cite{EskinInvHyp}. We have several cases to consider.

Let $\hat x$ be an interior point of the characteristic segment $\gamma \subseteq \partial \Omega$. If $\gamma$ belongs to the $\plus$ family then the construction of $\Omega_{r+p}$ shows that the curves of the $\minus$ family intersect $\gamma$ and directed inside $\Omega_{r+p}$ when $x_0$ increases. Since $\gamma$ is a characteristic curve it follows from \cite{EskinInvHyp} that $K_+(\hat x)$ is contained in either $\Pi_\nu^+$ or in $\Pi_\nu^-$. The tangent vector of the curve of the $\minus$ family passing through $\hat x$ is the projection onto $(x_1,x_2)$ of a forward null bicharacteristic and belongs to $\Pi_\nu^-$. Therefore $K_+(\hat x) \subseteq \Pi_\nu^-$, i.e.\ $\hat x$ is a point of no escape. 

If $\gamma$ is a characteristic curve of the $\minus$ family and $\widetilde  x \in \gamma$, then the $\plus$ family curve passing through $\widetilde x$ is the projection of a forward null bicharacteristic and its tangent vector at $\widetilde x$ belongs to $\Pi_{\nu_1}^-$, so again $K_+(\widetilde x) \subseteq  \Pi_{\nu_1}^-$, i.e.\ $\widetilde  x$ is also a point of no escape. Here $\nu_1$ is the exterior normal to $\gamma^{\scriptscriptstyle(1)}$ at $\widetilde  x$. Let $x^{\scriptscriptstyle(1)}$ be a corner point of $\partial \Omega_{r+p}$ at the intersection of characteristic segments $\gamma_+$, $\gamma_-$ belonging to the $\plus$, $\minus$ families, respectively. Let $\nu_+$, $\nu_-$ be the exterior normals to $\gamma_+$, $\gamma_-$ at the point $x^{\scriptscriptstyle(1)}$. As above we get that $K_+(x^{\scriptscriptstyle(1)}) \subseteq  \Pi_{\nu_+}^-$, $K_+(x^{\scriptscriptstyle(1)}) \subseteq  \Pi_{\nu_-}^-$, i.e.\ $K_+(x^{\scriptscriptstyle(1)}) \subseteq  \Pi_{\nu_+}^- \cap \Pi_{\nu_-}^-$. Therefore $x^{\scriptscriptstyle(1)} \in \partial \Omega_{r+p}$ is also a point of no escape, since any vector of $K_+(x^{\scriptscriptstyle(1)})$ points inside $\Omega_{r+p}$. 

Let now $\hat x \in \partial \Omega_{r+p}$ be a tangential point on $\partial \Omega$. It follows from \cite{EskinInvHyp} that either $K_+(\hat x) \subseteq \Pi^+_\nu$ or $K_+(\hat x) \subseteq \Pi^-_\nu$. Since $\partial \Omega$ is the ergosphere, $g_{00}(\hat x) = 0$ \cite{EskinInvHyp}. Thus $(\dot x_0,\dot x_1, \dot x_2) = (1,0,0)$ belongs to $\overline{K_+(\hat x)}$ since $\sum_{j,k=0}^2 g_{jk}(\hat x) \dot x_j \dot x_k = g_{00}(\hat x) = 0$. Therefore $\overline{K_+(\hat x)}$ is tangent to the plane $\dot x_1 \nu_1 + \dot x_2 \nu_2 = 0$. Here $\nu = (\nu_1,\nu_2)$ is the outward normal to $\partial \Omega_{r+p}$ at $\hat x$. 

Suppose for a moment that $K_+(\hat x) \subseteq \Pi^-_\nu$. Since $(1,0,0) \in \overline{K_+(\hat x)}$, for any small $\epsilon > 0$, $(1,\epsilon \dot x_1, \epsilon \dot x_2) \in K_+(\hat x)$ when $\dot x_1 \nu_1 + \dot x_2 \nu_2 < 0$, for arbitrary $(\dot x_1,\dot x_2)$. Therefore $K_+(\hat x) = \Pi^-_\nu$ when $\hat x$ is a tangential point. Similarly, if $K_+(\hat x) \subseteq \Pi^+_\nu$, then $K_+(\hat x) = \Pi^+_\nu$.

Let $\hat x_n \to \hat x$, where $\hat x$ is a tangential point in $\partial \Omega$, and each $\hat x_n \in \partial \Omega_{r+p}$ is an interior point of $\Omega$. The points $\hat x_n$ are no-escape points for $\partial \Omega_{r+p}$, as was proven above. Note that $\partial \Omega_{r+p}$ is smooth in a neighborhood of $\hat x$. Since $\hat x_n$ are no-escape points we have $K_+(\hat x_n) \subseteq \Pi^-_{\nu_n}$, where $\nu_n$ is the outward unit normal to $\partial \Omega_{r+p}$ at $\hat x_n$. We have $\Pi^-_{\nu_n} \to \Pi^-_\nu$, $K_+(\hat x_n) \to K_+(\hat x)$. Therefore $\overline{K_+(\hat x)} \subseteq \overline{\Pi^-_\nu}$, i.e.\  $\hat x$ is a point of no escape. 

\begin{remark} \label{rmk3.3}
These arguments hold for any characteristic point $\hat x \in \partial \Omega$ such that there exists a sequence $\hat x_n \to \hat x$ with $K_+(\hat x_n) \subseteq \Pi_{\nu_n}^-$. 

Suppose we have a characteristic segment $\subseteq \partial \Omega$. At the endpoints of the segment we have a sequence of points $\hat x_n$ as above. Thus the endpoints are no escape points. For any interior point of the segment we get that $K_+(\hat x) \subseteq \Pi^-_\nu$ by continuity. \xqed{\lozenge}
\end{remark}

\begin{remark}
Note that if $\hat x \in \partial \Omega$ is not tangential then it is an escape point: There exists a characteristic direction $\nu_0$ which is not normal to $\partial \Omega$. Since $g_{00}(\hat x) = 0$ we have that $K_+(\hat x)$ is either equal to $\Pi^-_{\nu_0}$ or to $\Pi^+_{\nu_0}$. In both cases there are directions 
of $K_+(\hat x)$ which point toward the exterior of $\Omega$.\xqed{\lozenge}
\end{remark}

Therefore we have proven:

\begin{lemma}\label{lemma3.5}
$\mathbb{R} \times \partial \Omega_{r+p}$ is a black hole event horizon.
\end{lemma}

\subsection{The case when $\partial \Omega$ has finitely many characteristic segments and finitely many characteristic points}

Suppose there are finitely many open intervals $L_1, \ldots , L_m$ in $\partial \Omega$, with $\overline{ L_j } \cap \overline{ L_k } = \emptyset$, $j \neq k$, such that the vector fields $f^\pm(x)$ are tangent to $\partial \Omega$ along $L_j$, $1 \leq j \leq m$. (Note that $f^+ = f^-$ on $\partial \Omega$.) Assume in addition that there are finitely many isolated tangent points $\beta_1, \ldots ,\beta_r$. 

We again let the open set $\Omega^+$ be as in Lemma \ref{twoposs}, $z_0 \in \partial \Omega^+$ an interior point of $\Omega$, and $\gamma_0$ a curve of the $\plus$ family passing through $z_0$, with endpoints $\alpha_1,\alpha_2 \in \partial \Omega$ (unless $\gamma_0$ is a closed orbit, in which case we are done), (cf. the second part of Lemma \ref{twoposs}). 

We claim that it is impossible to have $\alpha_1 \in L_{j_1}, \alpha_2 \in L_{j_2}$. If this is the case, consider neighborhoods $U(\alpha_1,\epsilon_1) \subseteq L_{j_1}, U(\alpha_2,\epsilon_2) \subseteq L_{j_2}$. For $\epsilon_1,\epsilon_2$ small there are solutions of the $\plus$ family $x^+_\alpha(x_0)$ that are close to $\gamma_0$ and have endpoints $\alpha \in U(\alpha_1,\epsilon_1), \widetilde  \alpha \in U(\alpha_2,\epsilon_2)$. Note that $L_j$ is an envelope for the $\plus$ family (see Remark \ref{rmk2.2}). All such solutions $x^+_\alpha(x_0)$ are not in $\Omega^+$, so $z_0 \not \in \partial \Omega^+$. 

Also, from the proof of Lemma \ref{twoposs}, it is impossible to have $\gamma_0 \subseteq  \partial \Omega^+$ which intersects $\partial \Omega$ transversally at both endpoints. Analogously, there is no $\gamma_0 \subseteq  \partial \Omega^+$ with one endpoint belonging to some $L_j$ and the other intersecting $\partial \Omega$ transversally.

Therefore $\gamma_0$ must have at least one endpoint either among $\beta_1, \ldots ,\beta_r$ or among the endpoints of $\overline{ L_1 },  \ldots , \overline{ L_m }$. Thus there are a finite number of such curves. Following the proof of Lemma \ref{twoposs} we get that the boundary of $\Omega^+$ consists of a finite number of characteristic segments inside $\Omega$ of the $\plus$ family, a finite number of the segments $\overline{L_j}$, $1 \leq j \leq m$ or closed subintervals of $\overline{ L_j }$ and a finite number of segments of $\partial \Omega$ where $\plus$ family trajectories start as $x_0$ increases. Starting with $\overline{ \Omega^+ }$ instead of $\overline{ \Omega }$ we consider the open set $\Omega_1^- \subseteq \Omega^+$ of $\minus$ family trajectories ending on $S_\epsilon$, and it is clear that we may repeat the proof of Lemma \ref{twoposs}. We get after a finite number of steps that the boundary $\Omega_1^-$ consists of a finite number of characteristic segments or parts of characteristic segments inside $\Omega$, some belonging to the $\plus$ family and some to the $\minus$ family and some characteristic segments that are parts of $\cup_{j=1}^m L_j$. It follows from the proof of Lemma \ref{lemma3.5} and Remark \ref{rmk3.3} that $\mathbb{R} \times \Omega^-$ is a black hole event horizon. Note that the boundary of $\partial \Omega_1^-$ may have corners -- i.e.\ it may only be piecewise smooth.

\subsection{The general case}

Consider $\Omega^+$. 
We have that $\overline{ \Omega^+ }$ does not intersect any of the open intervals $(\alpha_k,\beta_k)$ in $\partial \Omega$ where $\plus$ family curves end as $x_0$ increases. There can be at most countably many such intervals. Denote by $\Omega_k^+$ the union of all $\plus$ family curves ending on $(\alpha_k,\beta_k)$ as $x_0$ increases. Note that $\Omega_k^+ \cap \Omega^+ = \emptyset$. Take any $z_0 \in \partial \Omega_k^+$ which is an interior point of $\Omega$. Denote by $\gamma_0$ the $\plus$ family curve passing through $z_0$. Then $\gamma_0$ ends at either $\alpha_k$ or $\beta_k$, say $\alpha_k$ to fix ideas. Let $\alpha_{k_1}$ be a point on $\partial \Omega$ where $\gamma_0$ starts. Denote by $\Omega_{k_1}$ the domain bounded by $\gamma_0$ and $\partial \Omega$ and not containing $O$. If $\Omega_{k_1}$ contains $\Omega_k^+$ we replace $\Omega$ by $\Omega_1 = \Omega \less \overline{ \Omega_{k_1} }$. If $\Omega_{k_1}$ does not contain $\Omega_k^+$ then there is another characteristic curve $\gamma^{\scriptscriptstyle(0)}$ belonging to the boundary of $\Omega_k^+$ and ending at $\beta_k$. Let $\beta_{k_1} \in \partial \Omega$ be the starting point of $\gamma^{\scriptscriptstyle(0)}$. Let $\Omega_k^{\scriptscriptstyle(1)}$ be the domain bounded by $\gamma^{\scriptscriptstyle(0)}$ and $\partial \Omega$ that contains $\Omega_{k_1}$ and $\Omega_k^+$ and we shall replace $\Omega$ by $\Omega \less \overline{\Omega_k^{\scriptscriptstyle(1)}}$. Note that $\partial (\Omega \less \overline{\Omega_k^{\scriptscriptstyle(1)}})$ does not contain $(\alpha_k,\beta_k)$. Note also that $\partial(\Omega \less \overline{\Omega_k^{\scriptscriptstyle(1)}})$ is smooth at $\beta_k$ but may have a corner at $\beta_{k_1}$. In the latter case $\beta_{k_1}$ belongs to an open inverval $(\sigma,\delta)$ where the curves of the $\plus$ family start. Consider any other interval $(\alpha_j,\beta_j)$, $j \neq k$, where curves of the $\plus$ family end. Let $\Omega_j^{\scriptscriptstyle(1)}$ be a domain constructed as with $\Omega_k^{\scriptscriptstyle(1)}$. Since curves of the $\plus$ family do not intersect in $\Omega$ we have that $\Omega_j^{\scriptscriptstyle(1)} \cap \Omega_k^{\scriptscriptstyle(1)} = \emptyset$. Note that $\overline{\Omega_j^{\scriptscriptstyle(1)}} \cap \overline{\Omega_k^{\scriptscriptstyle(1)}}$ is either empty or consists of at most two tangential points in $\partial \Omega$. Denote $\Omega_\infty^+ = \cap_{k=1}^\infty (\Omega \less \overline{\Omega_j^{\scriptscriptstyle(1)}}) = \Omega \less \cup_{j=1}^\infty \overline{\Omega_j^{\scriptscriptstyle(1)}}$.


The boundary $\partial \Omega_\infty ^+$ consists of characteristic segments of the $\plus$ family, a closed set of tangent points belonging to $\partial \Omega$, and intervals $(\sigma_k,\delta_k)$, $k = 1,2, \ldots $, or parts of such intervals, where the $\plus$ family of curves start when $x_0$ increases. We shall show (cf. below) that $\partial \Omega_\infty^+$ is smooth, except possibly at a countable number of corner points $\beta_{k_j}$ belonging to some of the open invervals $(\sigma_k,\delta_k)$. Now consider the union $\Omega^-_k$ of all $\minus$ family curves in $\Omega$ that end on $(\sigma_k,\delta_k)$ when $x_0$ increases. Let $z_1 \in \partial \Omega_k^-$ be an interior point of $\Omega_\infty^+$ and let $\gamma_1^-$ be the $\minus$ family curve passing through $z_1$. Let $(\sigma_{k_1},\sigma_k)$ be the endpoints of $\gamma_1^-$. Consider also the $\minus$ family curve $\gamma_-^{\scriptscriptstyle(1)}$ ending at $\delta_k$ and belonging to $\partial \Omega_k^-$. Here, it is possible that $\gamma_-^{\scriptscriptstyle(1)}$ is a single point. Let $\Omega_k^{\scriptscriptstyle(2)}$ be the domain bounded by $\partial \Omega$ and either $\gamma_1^-$ or $\gamma_-^{\scriptscriptstyle(1)}$, which contains $\Omega_k^-$ and does not contain $O$. To fix ideas let $\gamma_1^- \subseteq \partial \Omega_k^{\scriptscriptstyle(1)}$. Then we replace $\Omega_\infty^+$ by $\Omega_\infty^+ \less \overline{\Omega_k^{\scriptscriptstyle(2)}}$. 

If we have $\beta_{k_j} \in (\sigma_k,\delta_k) \cap \partial \Omega_\infty^+$ then $\beta_{k_j} \not\in  \partial (\Omega_\infty^+ \less \Omega_k^{\scriptscriptstyle(2)})$ since $(\sigma_k,\delta_k) \subseteq \Omega_k^{\scriptscriptstyle(2)}$. Denote by $\gamma_1^{\scriptscriptstyle(1)}$ the intersection of $\gamma_1^-$ with $\partial \Omega_\infty^+$. Then the endpoints of $\gamma_1^{\scriptscriptstyle(1)}$ are either tangential points of $\partial \Omega$ or corner points of $\partial \Omega_\infty^+$ belonging to the interior of $\Omega$. 

Repeating this procedure for all $(\sigma_k,\delta_k)$, $k = 1,2,\ldots$, and for all characteristic segments $\gamma_j^+$ such that the $\minus$ family curves end on $\gamma_j^+$, we get a domain $\Omega_\infty^- \subseteq \Omega_\infty^+$ such that the boundary of $\Omega_\infty^-$ consists of characteristic segments belonging to either the $\plus$ or $\minus$ family and a closed set $S_1 \subseteq \partial \Omega \cap \partial \Omega_\infty^-$ of tangential points.

We shall show that $\partial \Omega_\infty^-$ is continuously differentiable except at corner points. It is enough to show that $\partial \Omega_\infty^-$ is continuously differentiable at any point of $S_1 = \partial \Omega_\infty^- \cap \partial \Omega$. Let $x^{(0)}$ be any point of $S_1$. Introduce $(\erho,\theta)$ coordinates in a small neighborhood $U_0$ of $x^{(0)} = (0,\theta_0)$. We have by (\ref{eqn2.7}),(\ref{eqn2.8}),
\begin{equation}
\label{eqn3.2}
\begin{aligned}
\frac{d\erho^\pm}{d\theta} 
&= \frac{\pm \sqrt{\erho_1}+g^{\erho\theta}(\erho,\theta)}{g^{\theta\theta}(\erho,\theta)}, \\
\frac{d\erho^\pm}{dx_0} 
&= \frac{\pm g^{\erho\theta}(\erho,\theta)\sqrt{\erho_1}-\erho_1}{b(\erho,\theta) \pm g^{0\theta}(\erho,\theta)\sqrt{\erho_1}},
\end{aligned}
\end{equation}
where $g^{\theta\theta}(0,\theta_0) < 0$, $b(0,\theta_0) < 0$, $g^{\erho\theta}(0,\theta_0) = 0$. Since $U_0$ is small we may assume that $g^{\theta\theta} < 0$, $b(\erho,\theta) \pm g^{0\theta}(\erho,\theta)\sqrt{\erho_1} < 0$ in $U_0$. In $U_0$, there are at most countably many intervals $(\alpha_k,\beta_k)$ where $g^{\erho\theta}(0,\theta) > 0$ or $g^{\erho\theta} < 0$. Let $(\alpha_1,\beta_1)$ be such that $g^{\erho\theta}(0,\theta) > 0$ on $(\alpha_1,\beta_1)$. It follows from (\ref{eqn3.2}) that curves of the $\plus$ family end on $\set{\erho = 0}$ when $x_0$ increases. We shall prove that there exists a curve $\erho = \erho_1(\theta)$ of the $\plus$ family starting at $\alpha_1$ and ending at $\beta_1$ such that $\erho = \erho_1(\theta)$ is the boundary of all curves of the $\plus$ family ending on $(\alpha_1,\beta_1)$. Let $w = \sqrt{\erho}$. We have $g^{\erho\theta}(w^2,\theta) = c_1(w^2,\theta)w^2 + g^{\erho\theta}(0,\theta)$. Since $U_0$ is small we have by the contraction mapping theorem that
\begin{align*}
-\sqrt{\erho_1} + g^{\erho\theta}(\erho,\theta) = c_2(\sqrt{\erho},\theta)(-\sqrt{\erho} + g_1(\theta)),
\end{align*}
where $c_2 > 0$, $g_1(\theta) > 0$, $\theta \in (\alpha_1,\beta_1)$. 

\begin{figure}
\begin{centering}
 \includegraphics[width=.6\linewidth]{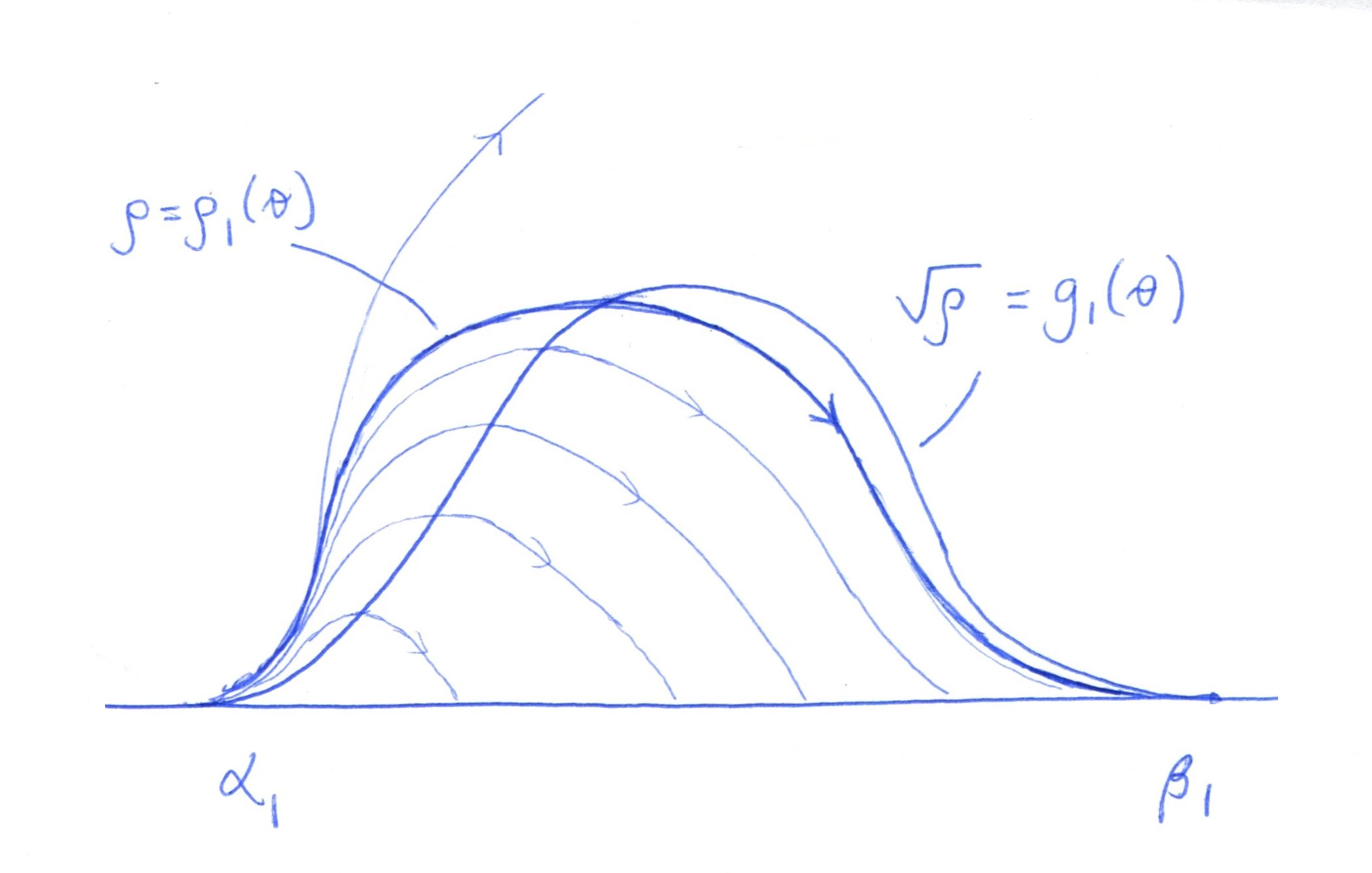}
 \caption{The curve $\rho = \rho_1(\theta)$ is the boundary of all curves of the $\plus$ family starting at $\alpha_1$ and ending on $(\alpha_1,\beta_1]$.}
 \label{figure}
\end{centering}
 \end{figure}

Consider the domain $V$ bounded by $w = g_1(\theta)$ and $w = 0$. We have that $\frac{d\erho}{d\theta} = 0$ when $w = g_1(\theta)$, $\frac{d\erho^+}{d\theta} < 0$ inside $V$ (since $g^{\theta\theta}< 0$) and $\frac{d\erho^+}{d\theta} > 0$ outside of $V$. Therefore curves $\erho = \erho_+(\theta)$ of the $\plus$ family that end at $(0,\theta')$, $\theta' \in (\alpha_1,\beta_1)$ increase when $\theta$ decreases until $\erho = \erho_+(\theta)$ intersects $\sqrt{\erho} = g_1(\theta)$. 
Then $\erho_+(\theta)$ decreases outside $V$ for $\alpha_1 < \theta < \beta_1$ when $\theta$ decreases. Note that $\erho = \erho_+(\theta)$ cannot cross the solution $\erho = (w_*^+(\theta))^2$ constructed in \ref{Lemma2.1}, since they belong to the same family. Therefore $\erho = \erho_+(\theta)$ must end at 
$\theta = \alpha_1$ (see Figure \ref{figure}).

Analogously if $(\alpha_2,\beta_2)$ is an interval in $U \cap \set{\erho = 0}$ such that $g^{\erho\theta}(0,\theta) < 0$, then there exists a $\minus$ family curve $\erho = \erho_2(\theta)$ that starts on $\beta_2$ and ends on $\alpha_2$ such that $\erho = \erho_2(\theta)$ is the boundary for all $\minus$ family curves that end at $(0,\theta)$, where $\theta \in (\alpha_2,\beta_2)$. 

Let $\erho = \erho(\theta)$ be a function on $U \cap \set{\erho = 0}$ equal to $\erho_k(\theta)$ on $(\alpha_k,\beta_k)$ and zero otherwise. The function $\erho = \erho(\theta)$ is the boundary of $\Omega_\infty^- \cap U$.

We shall show that $\erho = \erho(\theta)$ is continuously differentiable at any point $\partial \Omega_\infty^- \cap U$. Let $(0,\theta')$ be any point in $U_0$ such that $g^{\erho\theta}(0,\theta') = 0$. For any $\epsilon > 0$ there is $\delta > 0$ such that $|g^{\erho\theta}(0,\theta)| < \epsilon$ when $|\theta - \theta'| < \delta$. Let $(\alpha_j,\beta_j)$ be any interval in $(\theta'-\delta,\theta'+\delta)$ such that $|g^{\erho\theta}(0,\theta)| \neq 0$ on $(\alpha_j,\beta_j)$. We have
\begin{align*}
|\erho_j(\theta)| \leq \max_{[\alpha_j,\beta_j]} |g_j(\theta)| \leq C \max_{[\alpha_j,\beta_j]} |g^{\erho\theta}(0,\theta)| < C \epsilon.
\end{align*}

Therefore $|\erho(\theta)| < C\epsilon$ for $(\theta' - \delta,\theta' + \delta)$, i.e.\ $\lim_{ \theta \to \theta' } \erho(\theta) = 0$. This proves the continuity of $\erho(\theta)$. Analogously, $\left| \frac{d\erho(\theta)}{d\theta} \right| \leq C \left|g^{\erho\theta}(\erho(\theta),\theta) \right| + \sqrt{\erho} \leq C(|g^{\erho\theta}(0,\theta)| + \sqrt{\erho})$. Thus $\lim_{\theta \to \theta'} \frac{d\erho(\theta)}{d\theta} = 0$, i.e.\ $\frac{d\erho(\theta)}{d\theta}$ is also continuous. 

As in Lemma \ref{lemma3.5} and Remark \ref{rmk3.3}, we get that any point of $\partial \Omega_\infty^-$ is a no-escape point, i.e.\ $\mathbb{R} \times \Omega_\infty^-$ is a black hole.

\begin{remark}
The black hole constructed in this subsection may be different from the black holes constructed in the previous subsections, in the case when there is more than one black hole.\cite{EskinNonstationary} \xqed{\lozenge}
\end{remark}

\begin{remark}\label{remark3.7}
At tangential points, $\partial \Omega_\infty^-$ is $C^1$ but not $C^2$ in general, since there are characteristic curves of different families that have a common tangential point. \xqed{\lozenge}
\end{remark}

\section{Acoustic metrics and an example with corners}

\label{Acoustic}

\subsection{Acoustic metrics}

\label{Smooth}

We consider acoustic waves in a moving medium. The \emph{acoustic metric} associated to a vector field $v = (v_1,v_2)$ is the (stationary) Lorentzian metric $\frac{\rho}{c}[c^2 dx_0^2 - (dx - v dx_0)^2]$, i.e.\ the metric $g$ given by
\begin{align}\label{metric}
g_{00} = \frac{\rho}{c}(c^2-|v|^2), \quad g_{0j} = g_{j0} = \frac{\rho}{c}v_j, 1 \leq j \leq 2, \quad g_{ij} = -\frac{\rho}{c}\delta_{ij}, 1 \leq i,j \leq 2,
\end{align}
The inverse of the metric tensor is given by
\begin{align*}
g^{00} = \frac{1}{\erho c}, \quad g^{j0} = g^{0j} = \frac{1}{\rho c}v_j, 1 \leq j \leq 2, \quad g^{jk} = \frac{1}{\rho c}(v_jv_k - c^2\delta_{jk}), 1 \leq j,k \leq 2.
\end{align*}
We assume that the flow $v = (v_1,v_2)$ is irrotational, i.e. there exists a potential $\psi$ such that $v = \grad \psi$, barotropic, i.e. $p = p(\rho)$ where $p$ is the pressure and $\erho$ is the density. Moreover, $v$ and $\rho$ satisfy the continuity equation
\begin{align*}
\rho_t + \grad \cdot (\rho \grad \psi)  = 0,
\end{align*}
and the Euler equation, which can be reduced to the form \cite{VisserAcoustic}
\begin{align*}
\psi_t + h + \frac{1}{2}(\nabla \psi)^2 + \Phi = 0, 
\end{align*}
where $\Phi$ represents external forces and $h(p)$ is the specific enthalpy.

In the case when $v$ and $\rho$ satisfy these requirements, the wave equation (\ref{eqn1.1}) for a metric of the form (\ref{metric}) is a physical model for the propagation of sound waves (see \cite{VisserAcoustic}) where $c = \sqrt{\frac{dp}{d\rho}}$ is is the speed of sound.

We shall take $\rho$ to be constant. Then $p$ and $c$ are constant as well. Then by continuity equation
\begin{align*}
\Delta \psi = 0
\end{align*}
i.e. $\psi$ is a harmonic function. Rescaling, we shall assume that $c = 1$. Then the ergoregion is where $1 - |v|^2 < 0$. 

\begin{remark}
Other well-known spacetime metrics may be transformed into the form (\ref{metric}) after an appropriate choice of coordinates, including the Schwarzschild metric in Painlev\' e-Gullstrand coordinates \cite{VisserAcoustic}.\xqed{\lozenge}
\end{remark}

\begin{remark}
As was noted in the introduction, acoustic metrics are not the only physical examples of analogue (articifial) black holes. There are models for optical black holes, surface waves, relativistic acoustic waves, Bose-Einstein condensates, and others. See the references in the introduction.\xqed{\lozenge}
\end{remark}

It will be convenient to write the vector field in polar coordinates as $v = v_r \hat r + v_\theta \hat \theta$, $v_r = \frac{\partial \psi}{\partial r}$, $v_\theta = \frac{1}{r}\frac{\partial \psi}{\partial \theta}$. In this case the vector field is a solution of the Euler equations. We will specify an explicit choice of $\psi$ in the following subsection.

Let
$$v = \frac{A(r,\theta)}{r} \hat r + \frac{B(r,\theta)}{r} \hat \theta, \quad A,B \in C^\infty,$$ 
and let $g$ be the corresponding acoustic metric, which satisfies (\ref{eqn1.2})-(\ref{eqn1.4}). In polar coordinates, the form corresponding to (\ref{charnull}) is
\begin{align*}
\left(\frac{A^2}{r^2} - 1\right)\xi_r^2 + 2\frac{AB}{r^3}\xi_r\xi_\theta + \left(\frac{B^2}{r^4} - \frac 1{r^2} \right)\xi_\theta^2 = 0,
\end{align*}
i.e. $g^{rr} = \frac{A^2}{r^2}-1$, $g^{r\theta} = g^{\theta r} = \frac{AB}{r^3}$, $g^{\theta\theta} = \frac{B^2}{r^4} - \frac{1}{r^2}$. We find the solutions
\begin{align*}
\xi_\theta^\pm = \frac{ -\frac{AB}{r} \pm \sqrt{\erho}}{\frac{B^2}{r^2}-1}
\xi_r^\pm .
\end{align*}
In addition, the acoustic metric satisfies $g^{r0} = \frac{A}{r}$, $g^{\theta 0} = \frac{B}{r^2}$. Therefore the system (\ref{eqn2.4}) becomes 
\begin{equation}\label{eqn4.2}
\begin{aligned}
\frac{dr^\pm}{dx_0} 
&= \frac{g^{rr}f_2^\pm - g^{r\theta}f_1^\pm}{g^{r0}f_2^\pm - g^{\theta 0}f_1^\pm} 
= \frac{\left( \frac{A^2}{r^2} - 1 \right) f_2^\pm - \frac{AB}{r^3}f_1^\pm}{\frac{A}{r}f_2^\pm - \frac{B}{r^2}f_1^\pm}\\
\frac{d\theta^\pm}{dx_0} 
&= \frac{g^{\theta r}f_2^\pm - g^{\theta\theta}f_1^\pm}{g^{r0}f_2^\pm - g^{\theta 0}f_1^\pm}
= \frac{\frac{AB}{r^3} f_2^\pm - (\frac{B^2}{r^4}-\frac{1}{r^2})f_1^\pm}{\frac{A}{r}f_2^\pm - \frac{B}{r^2}f_1^\pm}.
\end{aligned}
\end{equation}
Near the ergosphere $A^2 + B^2 - r^2 = \erho = 0$, we can use 
\begin{equation}\label{eqn4.3}
\begin{aligned}
f_1^\pm &= \frac{AB}{r} \mp \sqrt{\erho} \\
f_2^\pm &= \frac{B^2}{r^2} - 1.
\end{aligned}
\end{equation}

\begin{remark}
Alternatively,  formulas for $f^\pm$ which are valid on all of $\Omega$, up to removable singularities, are
\begin{equation}
\begin{aligned}
f_1^\pm &= \frac{(A^2-r^2)(B \mp r)}{\frac{AB}{r} \pm \sqrt{A^2+B^2-r^2}}\\
f_2^\pm &= B \mp r.
\end{aligned}
\end{equation}\xqed{\lozenge}
\end{remark}
Denote

\begin{equation}\label{eqn4.4}
b_0 = \frac{A}{r}\left(\frac{B^2}{r^2}-1\right) - \frac{B}{r^2}\left(\frac{AB}{r}\mp \sqrt{\erho}\right) = - \frac{A}{r} \pm \frac{B}{r^2}\sqrt{\erho}
\end{equation}
Note that $b_0 > 0$ near $\erho = 0$ since $A < 0$. Therefore
\begin{equation}\label{eqn4.5}
\begin{aligned}
\frac{dr^\pm}{dx_0} 
&= \frac{( \frac{A^2}{r^2} - 1)( \frac{B^2}{r^2} - 1) - \frac{AB}{r^3}\left( \frac{AB}{r}\mp \sqrt{\erho}\right)}{b_0} 
= \frac{\pm\left(\frac{AB}{r} \mp \sqrt{\erho}\right)\sqrt{\erho}}{b_1} \\
\frac{d\theta^\pm}{dx_0} 
&= \frac{\frac{AB}{r^3}( \frac{B^2}{r^2} - 1) - ( \frac{B^2}{r^4} - \frac{1}{r^2})( \frac{AB}{r}\mp \sqrt{\erho})}{b_0} 
= \frac{\pm(\frac{B^2}{r^2}-1)\sqrt{\erho}}{b_1}.
\end{aligned}
\end{equation}
where $b_1 = r^2b_0 = -Ar \pm B \sqrt{\erho}$. 

For later, we record that in $(\erho,\theta)$ coordinates, we have
\begin{equation}\label{eqn4.6}
\begin{aligned}
\frac{d\erho^\pm}{dx_0} 
&= 2(AA_\theta + BB_\theta) \frac{d\theta^\pm}{dx_0} + 2(AA_r + BB_r - 2r)\frac{dr^\pm}{d\theta} \\
&= \frac{\pm 2(AA_\theta+ BB_\theta) ( \frac{B^2}{r^2}-1)\sqrt{\erho}}{b_1} + 2(AA_r + BB_r - r) \frac{\pm\left(\frac{AB}{r} \mp \sqrt{\erho}\right)\sqrt{\erho}}{b_1} \\
&= \pm \frac{2Q\sqrt{\erho}}{b_1} + \frac{2(r - AA_r - BB_r)\erho}{b_1}.
\end{aligned}
\end{equation}
where 
\begin{align}\label{eqn4.7}
Q = (AA_\theta + BB_\theta)\left(\frac{B^2}{r^2}-1\right) + (AA_r + BB_r-r)\frac{AB}{r}.
\end{align}
Since $b_1 > 0$, and $\frac{B^2}{r^2} - 1 < 0$ near $\erho = 0$, we have that $\frac{d\theta^\pm}{dx_0} \lessgtr 0$, i.e.\ $\theta^+(x_0)$ decreases and $\theta^-(x_0)$ increases when $x_0$ increases. We have
\begin{equation}\label{eqn4.8}
\begin{aligned}
\frac{d\erho^\pm}{d\theta} 
&= \frac{2Q}{\frac{B^2}{r^2}-1} \mp \frac{2(AA_r + BB_r - r) \sqrt{\erho}}{\frac{B^2}{r^2}-1}.
\end{aligned}
\end{equation}
Therefore $\erho^\pm = \erho^\pm(\theta)$ is tangential to $\erho = 0$ if and only if $Q = 0$.

It follows from (\ref{eqn4.6}) that near $\erho = 0$, $\frac{d\erho^+}{dx_0}< 0$ when $Q < 0$ and $\frac{d\erho^+}{dx_0} > 0$ when $Q > 0$. Therefore $(\erho^+(x_0),\theta^+(x_0))$ ends on $\erho = 0$ when $Q < 0$ and $(\erho^+(\theta_0),\theta^+(x_0))$ starts on $\erho = 0$ when $Q > 0$. Similarly
 $(\erho^-(x_0),\theta^-(x_0))$ starts on $\erho = 0$ when $Q < 0$ and ends on $\erho = 0$ when $Q > 0$.

\subsection{Example of an acoustic black hole with a corner}\label{sb4.2}

Consider a potential
$$
\psi = A_0 \log r + \epsilon r\sin\theta, \qquad A_0 < -1,\ 0 < \epsilon < 1,
$$
so that
$$
A = r\frac{\partial\psi}{\partial r} = A_0 + \epsilon r\sin\theta, \quad B = \frac{\partial \psi}{\partial \theta} =  \epsilon r\cos\theta.
$$
In $(w,\theta)$ coordinates, from (\ref{eqn4.5}),(\ref{eqn4.6}) we have
\begin{equation}\label{eqn4.9}
\begin{aligned}
\frac{dw^\pm}{dx_0} 
&= \frac{\pm Q + (r-(A_0 + \epsilon r\sin\theta)\epsilon\sin\theta - (\epsilon r\cos\theta) 
\epsilon \cos\theta) w}{-(A_0+\epsilon r\sin\theta) r \pm (\epsilon r\cos\theta) w} 
\\
\frac{d\theta^\pm}{dx_0}
&= \frac{\pm ((\epsilon \cos\theta)^2-1)w}{-(A_0+\epsilon r\sin\theta) r \pm (\epsilon r\cos\theta) w} .
\end{aligned}
\end{equation}
where 
\begin{align*}
Q &= [(A_0 + \epsilon r\sin\theta)\epsilon r\cos\theta + (\epsilon r\cos\theta)(-\epsilon r\sin\theta)](\epsilon^2 \cos^2\theta-1) \\
& \quad + 
[(A_0 + \epsilon r\sin\theta)\epsilon r\sin\theta + (\epsilon r\cos\theta)^2 - r^2](A_0 + \epsilon r\sin\theta)\epsilon r\cos\theta/r^2\\
&= \epsilon \cos\theta( A_0^2\epsilon\sin\theta + r^2\epsilon(\epsilon^2-1)\sin\theta + 2A_0 r(\epsilon^2-1)).
\end{align*}
The equation of the ergosphere $w = 0$ is $(A_0+\epsilon r\sin\theta)^2 + (\epsilon r\cos\theta)^2 - r^2 = 0$, which gives
\begin{align*}
r = r_0(\theta)
&= \frac{A_0 \epsilon \sin \theta + \sqrt{A_0^2 \epsilon^2 \sin^2 \theta + A_0^2(1-\epsilon^2)}}{1 - \epsilon^2}\\
&= \frac{-A_0}{1-\epsilon^2}(-\epsilon \sin \theta + \sqrt{1 - \epsilon^2 \cos^2 \theta}).
\end{align*}
Note that $r_0(\pi/2) = \frac{-A_0}{1+\epsilon}$, $r_0(-\pi/2) = \frac{-A_0}{1-\epsilon}$, and $\frac{-A_0}{1+\epsilon} \leq r(\theta) \leq \frac{-A_0}{1-\epsilon}$ 
for all $\theta$. 

When $w = 0$, we have $Q = -2A_0 (\epsilon r\cos\theta) (A_0 + \epsilon r\sin\theta)$. 
Thus there are tangential points where $w = 0$ and $\theta = \pm \frac{\pi}{2}$. If $w = 0$ and $\theta \neq \pm \pi/2$, 
we can only have tangential points when $A_0  + \epsilon r\sin\theta = 0$ and hence $(\epsilon r\cos\theta)^2 = r^2$, which is impossible when $|\epsilon| < 1$. 

\begin{itemize}
\item At the point $w = 0$, $\theta = \pi/2$, the linearization in $(w,\theta)$ has the Jacobian  matrix
\begin{align*}
\begin{bmatrix} 
-\frac{(1+\epsilon)^2}{A_0} & \mp 2\epsilon(1+\epsilon) \\
\mp \frac{(1+\epsilon)^2}{A_0^2} & 0
\end{bmatrix}
\end{align*}
which has determinant $-2\epsilon(1+\epsilon)^3/A_0^2 < 0$. Therefore $w= 0$, $\theta = \pi/2$ is a saddle point.  

\item At the points $w = 0$, $\theta = -\pi/2$, the linearization in $(w,\theta)$ has the Jacobian matrix
\begin{align*}
\begin{bmatrix}
-\frac{(1-\epsilon)^2}{A_0} & \pm 2\epsilon(1-\epsilon) \\
\mp \frac{(1-\epsilon)^2}{A_0^2} & 0
\end{bmatrix}
\end{align*}
which has determinant $2\epsilon(1-\epsilon)^3/A_0^2 > 0$, trace $-(1-\epsilon)^2/A_0 > 0$, and discriminant $(1-\epsilon)^4/A_0^2 - 8\epsilon(1-\epsilon)^3/A_0^2 = (1-\epsilon)^3(1-9\epsilon)/A_0^2$. Therefore $w = 0$, $\theta = -\pi/2$ is an unstable node for $0 < \epsilon < \frac{1}{9}$ and an unstable spiral for $\frac{1}{9} < \epsilon < 1$.  
\end{itemize}

In the next subsection we will show that from these calculations we can conclude that the black hole has a corner whenever the second critical point is a spiral.

\begin{figure}
 \centering
 \begin{subfigure}{.5\textwidth}
   \centering
   \includegraphics[width=.7\linewidth]{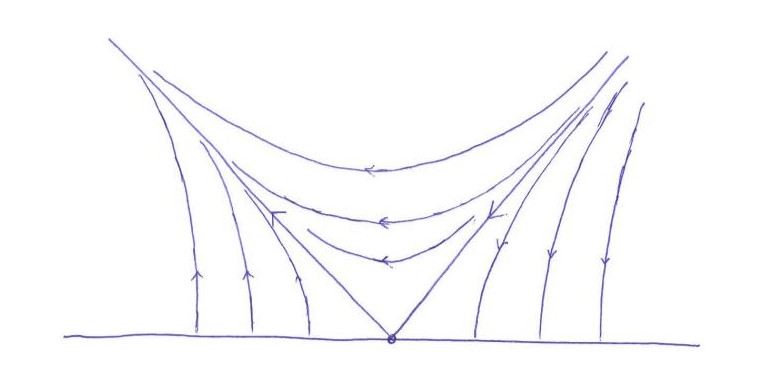}
   \caption{$(w,\theta)$ coordinates}
   \label{saddlerho}
 \end{subfigure}%
 \begin{subfigure}{.5\textwidth}
   \centering
   \includegraphics[width=.7\linewidth]{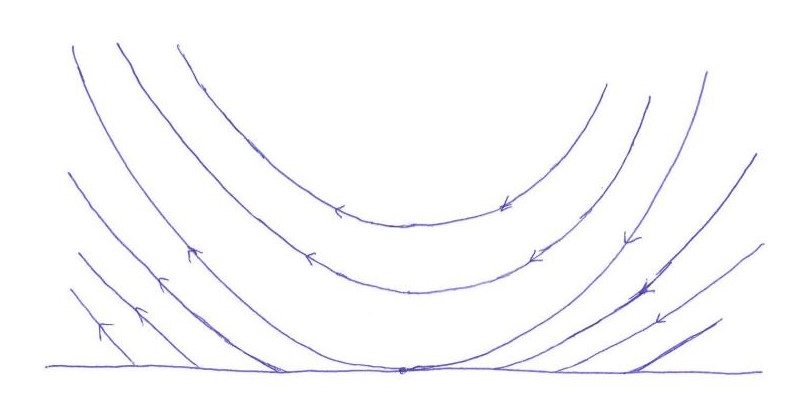}
   \caption{$(\erho,\theta)$ coordinates}
   \label{saddler}
 \end{subfigure}
 \caption{
The qualitative picture near
 a saddle point.}
 \label{fig2}
\end{figure}
\begin{figure}
 \begin{subfigure}{.5\textwidth}
   \centering
   \includegraphics[width=.7\linewidth]{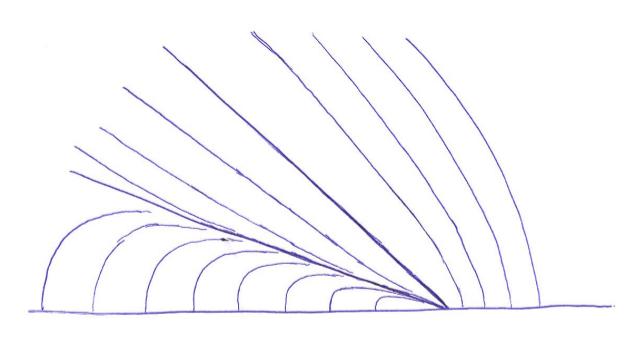}
   \caption{$(w,\theta)$ coordinates}
   \label{noderho}
 \end{subfigure}
 \begin{subfigure}{.5\textwidth}
   \centering
   \includegraphics[width=.7\linewidth]{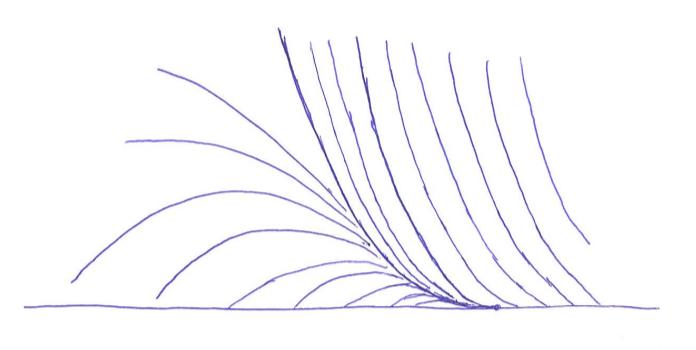}
   \caption{$(\erho,\theta)$ coordinates}
   \label{noder}
 \end{subfigure}
 \caption{
The qualitative picture near a node. }
 \label{fig3}
 \end{figure}
 \begin{figure}
 \begin{subfigure}{.5\textwidth}
   \centering
   \includegraphics[width=.9\linewidth]{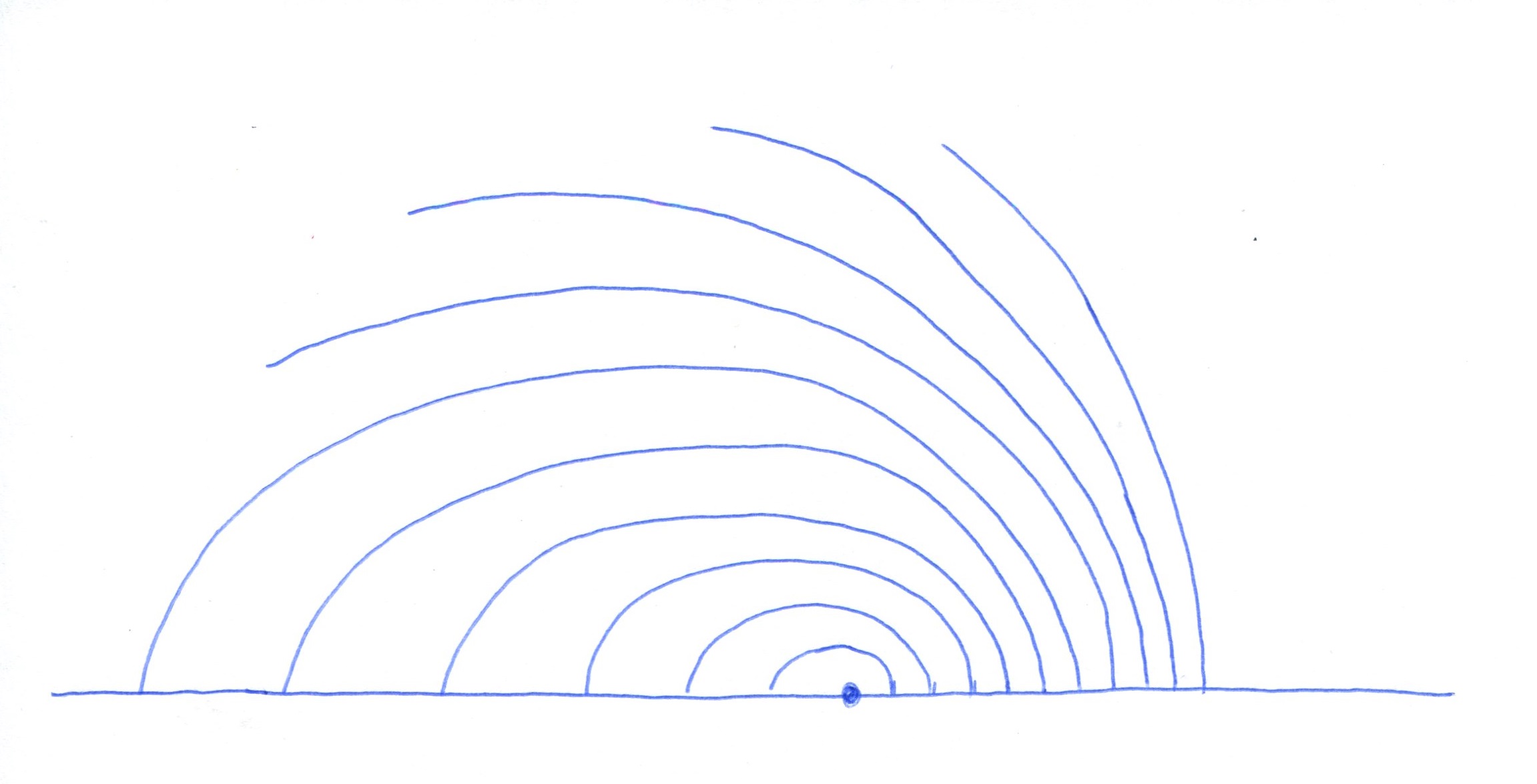}
   \caption{$(w,\theta)$ coordinates}
   \label{spiralrho}
 \end{subfigure}
 \begin{subfigure}{.5\textwidth}
   \centering
   \includegraphics[width=.9\linewidth]{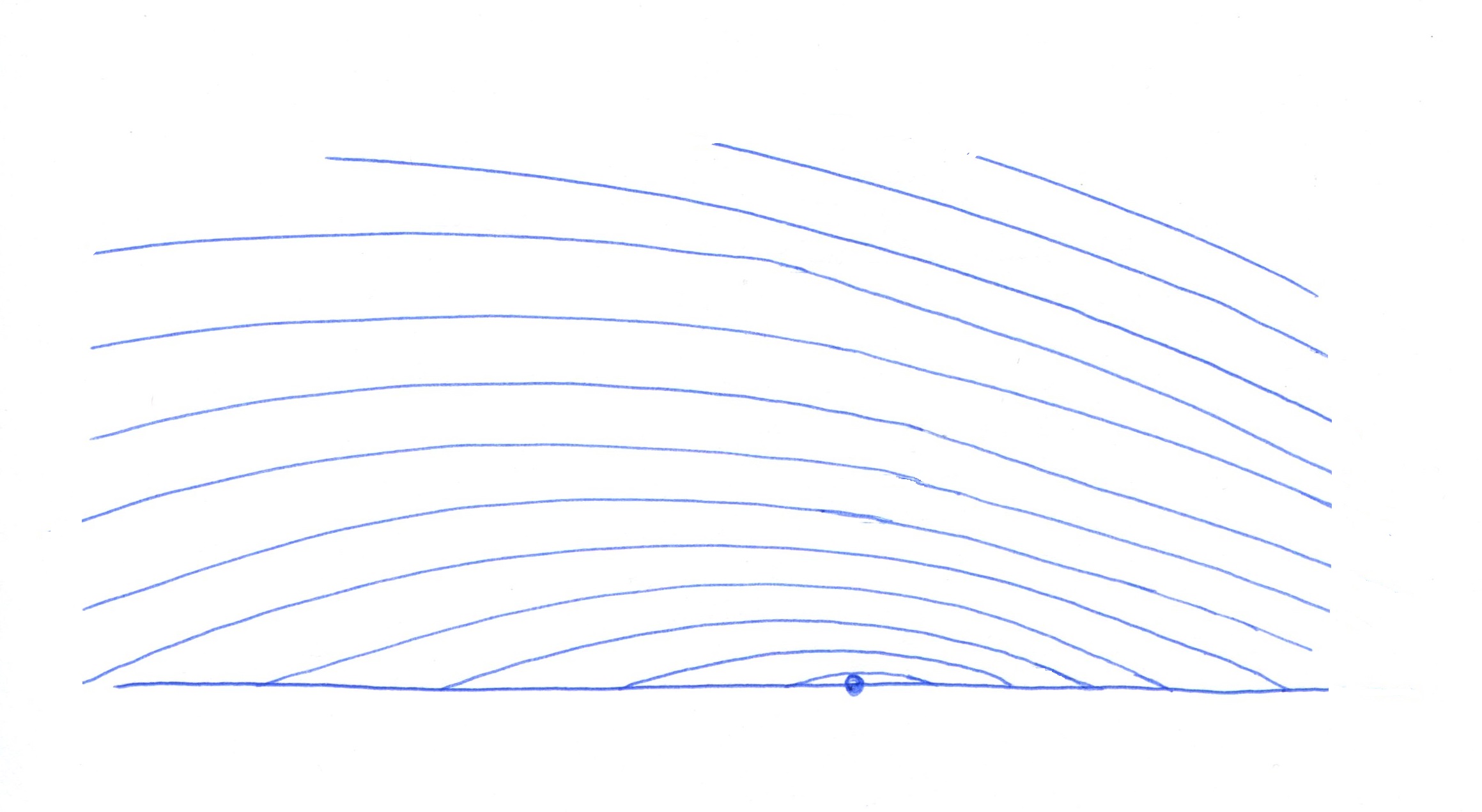}
   \caption{$(\erho,\theta)$ coordinates}
   \label{spiralr}
 \end{subfigure}
 \caption{
The qualitative picture near a spiral.}
 \label{fig4}
 \end{figure}

 \begin{figure}
 \centering
 \begin{subfigure}{.5\textwidth}
   \centering
   \includegraphics[width=.9\linewidth]{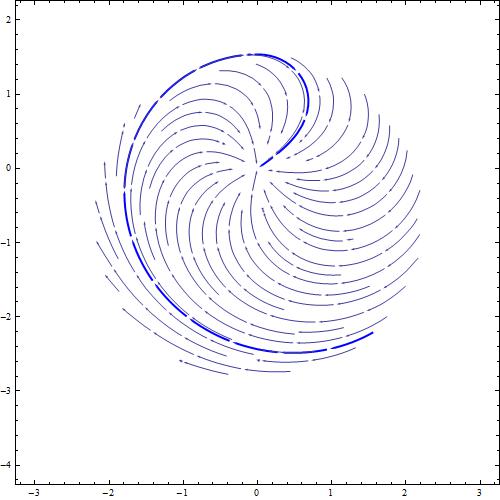}
   \caption{Trajectories for the $\plus$ family.}
   \label{Epsp3}
 \end{subfigure}%
 \begin{subfigure}{.5\textwidth}
   \centering
   \includegraphics[width=.9\linewidth]{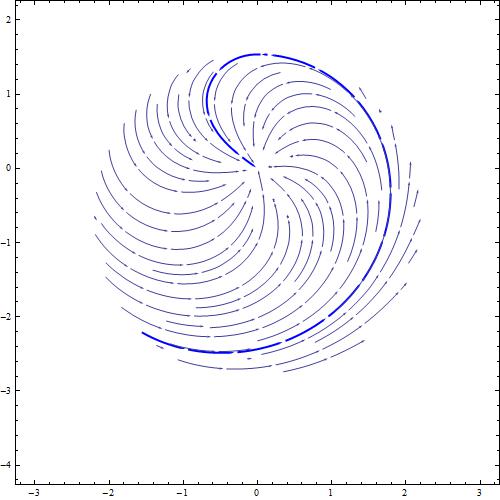}
   \caption{Trajectories for the $\minus$ family.}
   \label{Epsp3minus}
 \end{subfigure}
 \caption{Numerically plotted trajectories for (\ref{eqn4.5}) with $A = A_0 + \epsilon r \sin\theta$, $B = \epsilon r \cos\theta$, $A_0 = -2.0$, $\epsilon = 0.3$. The bold trajectories pass through $\theta = -\pi/2$, $r = 2.4350096$.}
 \label{fig1}
 \end{figure}

\subsection{Phase portrait with two critical points}

In this subsection we describe the generic phase portrait when there are two critical points. 
\subsubsection{One saddle and one spiral}
Consider first the case of one saddle point $\alpha_1 = \{ \erho =0,\ \theta = \pi/2 \}$ and one spiral $\alpha_2 = \{ \erho = 0,\ \theta = -\pi/2 \}$. Let us assume, to fix ideas, that that point $\alpha_1 = \{ \erho = 0,\ \theta = \pi/2 \}$ is a saddle, the point $\alpha_2 = \{ \erho = 0,\ \theta = -\pi/2 \}$ is an unstable spiral and the $\plus$ trajectories end on $\{ \erho = 0, \ 3\pi/2 < \theta < \pi/2 \}$ and start on $\{ \erho = 0,\ -\pi/2< \theta < \pi/2 \}$ when $x_0$ increases. Note that $\theta = -\pi/2 = 3\pi/2 \pmod{2\pi}$ is the same point. 

The $\plus$ trajectory $\gamma^+$ that ends at $\alpha_1$ must start at some point $\{ \erho = 0,\ \theta = \theta^+\}$ where $-\pi/2 < \theta^+ < \pi/2$. The $\plus$ trajectories starting on $\{ \erho = 0,\ -\pi/2 < \theta < \theta^+\}$ must end on $\{ \erho = 0,\ \pi/2 < \theta < 3\pi/2 \}$. The $\plus$ trajectories starting on $\{ \erho =0,\ \theta^+ < \theta \leq \pi/2\}$ must approach $O$ when $x_0$ increases. Therefore the set $\Omega^+$ of all $\plus$ trajectories ending at $O$ is bounded by $\gamma^+$ and $\{ \erho = 0,\ \theta^+ \leq \theta \leq \pi/2\}$. Analogously there exists a $\minus$ trajectory $\gamma^-$ that ends at $\alpha_1$ and starts at some point $\{ \erho = 0,\ \theta = \theta^-\}$ with $\pi/2 < \theta^- < 3\pi/2$. The set $\Omega^-$ of all $\minus$ trajectories ending at $O$ is bounded by $\gamma^-$ and $\{ \erho = 0,\ \pi/2 \leq \theta \leq \theta^- \}$. Thus the black hole $\Omega_0 = \Omega^+ \cap \Omega^-$ is bounded by segments of $\gamma^+$ and $\gamma^-$ which meet at a corner point.
The numerically computed phase portraits in Figure \ref{fig1} for $A = A_0 + \epsilon r \sin\theta$, $B = \epsilon r \cos\theta$, with $A_0 = -2.0$ and $\epsilon = 0.3$, indicate trajectories approximating $\gamma^+$ and $\gamma^-$ as described above. Combining the pictures in Figure \ref{Epsp3} and Figure \ref{Epsp3minus} we get a black hole with a corner.
 See also Figure \ref{fig5}.

\begin{figure}
\begin{centering}
  \includegraphics[width=.4\linewidth]{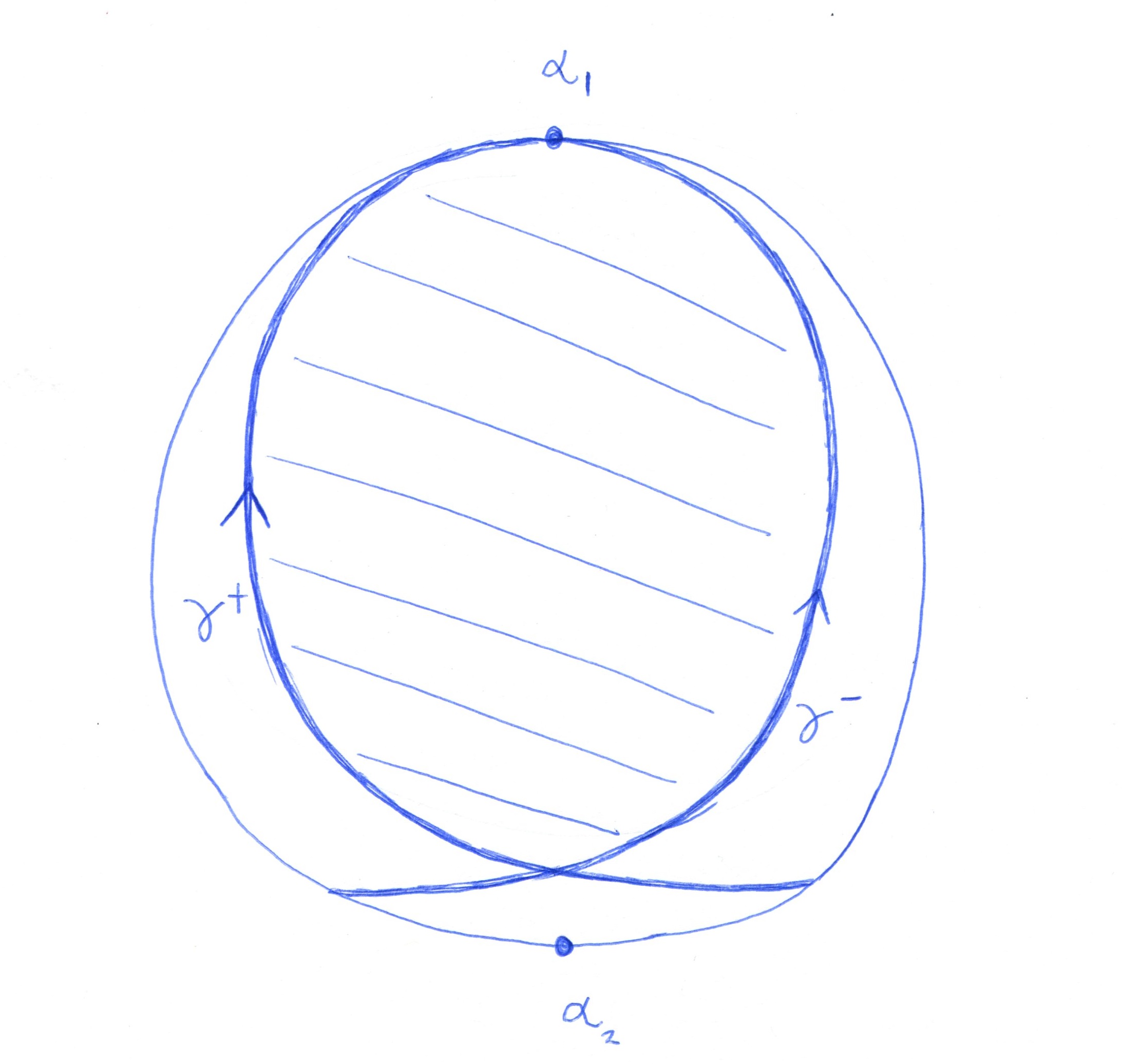}
  \caption{Qualitative sketch of a black hole with a corner in the case of two critical points.}
  \label{fig5}
\end{centering}
\end{figure}

%

\subsubsection{One saddle and one node}

Now we consider the slightly more difficult case where there is one saddle point and one node.  As before we assume that the $\plus$ trajectories end on $\{ \erho = 0,\ \pi/2 < \theta < 3\pi/2 \}$ and start on $\{ \erho = 0,\ -\pi/2 < \theta < \pi/2 \}$, and let assume that $\alpha_1 = \{ \erho = 0,\ \theta = \pi/2\}$ is a saddle point and $\alpha_2 = \{ \erho = 0,\ \theta = 3\pi/2\}$ is an unstable node. 

Consider all $\plus$ trajectories that start at the node $\alpha_2$. There are two cases. 

In the first case, the endpoints of the $\plus$ trajectories starting at the node cover the interval $\{ \erho = 0,\ \pi/2 < \theta < 3\pi/2\}$ of the ergosphere. It follows that there is a $\plus$ trajectory $\gamma_1^+$ starting at the node $\alpha_2$ and ending at the saddle point $\alpha_1$. More precisely, $\gamma_1^+$ approaches the node when $x_0 \to -\infty$ and approaches the saddle when $x_0 \to +\infty$. There can be $\plus$ trajectories emerging from the node that do not end on $\{ \erho = 0, \ \pi/2 \leq \theta \leq 3\pi/2 \}$. These trajectories must all end at the singularity $O$. Also, all $\plus$ trajectories starting on $\{ \erho = 0,\ -\pi/2 < \theta < \pi/2 \}$ end at $O$.  Therefore, the set $\Omega^+$ of all trajectories that end at $O$ is bounded by $\gamma_1^+$ and the part of the ergosphere $\{ \erho = 0,\ -\pi/2 \leq \theta \leq \pi/2\}$.

In the second case there exists $\pi/2 < \theta_1^+ < 3\pi/2$ such that the endpoints of the $\plus$ trajectories starting at the node cover the interval $\{ \erho = 0, \ \theta_1^+ \leq \theta < 3\pi/2\}$ of the ergosphere. Therefore there is a $\plus$ trajectory $\gamma_2^+$ ending at the saddle point $\alpha_2$ that starts at some point $\{ \erho = 0,\ \theta = \theta_2^+\}$, where $-\pi/2 < \theta_2^+ < \pi/2$. All $\plus$ trajectories starting on $\{ \rho = 0,\ -\pi/2 < \theta < \theta_2^+\}$ end on $\{ \erho = 0,\ \theta_1^+ < \theta < \pi/2 \}$ and all $\plus$ trajectories starting on $\{ \erho =0,\ \theta_2^+ < \theta < \pi/2 \}$ end at the singularity $O$, including the $\plus$ trajectory starting at $\alpha_2 = \{\erho =0,\ \theta = \pi/2 \}$. Therefore the set $\Omega^+$ of all $\plus$ trajectories approaching $O$ is bounded by $\gamma_2^+$ and the part of the ergosphere $\{ \erho =0,\ \theta_2^+ < \theta < \pi/2 \}$. 

For the $\minus$ trajectories there are also two cases. In one case there is a $\minus$ trajectory $\gamma_1^-$ that starts at some point $\{ \erho =0,\ \theta = \theta_2^-\}$, where $\pi/2 < \theta_2^- < 3\pi/2$, and ends at the saddle point $\alpha_1$. The set $\Omega^-$ of $\minus$ trajectories ending at $O$ is bounded by $\gamma_2^-$ and the part of the ergosphere $\{ \erho = 0,\ \pi/2 < \theta < \theta_2^-\}$. In the other case $\Omega^-$ is bounded by a $\minus$ trajectory $\gamma_2^-$ starting at $\alpha_2$ and ending at $\alpha_1$, and by the part of the ergosphere $\{\erho = 0,\ \pi/2 < \theta < 3\pi/2 \}$. The black hole $\Omega_0$ is the intersection of $\Omega^+$ and $\Omega^-$. Therefore $\Omega_0$ is bounded by (parts of) $\gamma_1^+$ or $\gamma_2^+$ or $\gamma_1^-$ or $\gamma_2^-$. Only in the case when $\Omega_0$ is bounded by $\gamma_1^+$ and $\gamma_1^-$ is the boundary $\partial \Omega_0$ smooth. In the three other cases $\partial \Omega_0$ has a corner points. We do not present numerical investigations of this case.

As in Remark \ref{remark3.7}, we note that even when $\partial \Omega_0$ is smooth it is $C^1$ but may not be $C^2$ since $\partial \Omega_0$ consists of two smooth curves $\gamma_1^+,\gamma_2^+$ tangential to the ergosphere at $\alpha_1$ and $\alpha_2$ and belonging to different families. 

\section{Determination of black holes by boundary measurements}

Let
\begin{equation}\label{eqn5.1}
Lu=0
\end{equation}
be the equation  (\ref{eqn1.1}) in the cylinder $\mathbb R\times\mathcal D$,  where 
$\D\subseteqq\R^2$
is a bounded domain with smooth  boundary  $\partial\D$  such that the ergoregion $\Omega=\{g_{00}(x)<0\}$  is contained  inside $\D$.   
Consider  the initial-boundary  value problem  for (\ref{eqn5.1})  in $\R\times\D$  with the boundary and initial conditions
\begin{align}
\label{eqn5.2}
&u\big|_{\R\times\partial D}=f,
\\
\label{eqn5.3}
&u=0\ \ \mbox{for}\ \ x_0\ll 0,\ \ x\in \D,
\end{align}
where $f$  has compact  support  in $\R\times\D$.  Let  $\Lambda$  be the Dirichlet-to-Neumann (DN)  operator  on  $\R\times\partial\D$,  i.e.
\begin{equation}\label{eqn5.4}
\Lambda f=\sum_{j,k=0}^n g^{jk}(x)\frac{\partial u}{\partial x_j}\nu_k(x)\Big(\sum_{p,r=0}^n g^{pr}(x)\nu_p(x)\nu_r(x)\Big)^{-\frac{1}{2}}\Big|_{R\times\partial\D},
\end{equation}
where  $n=2,\ \nu=(\nu_1,\nu_2)$  is  the outward  unit  normal  to  $\partial\D$,  and  $u=u(x_0,x)$  is the solution  of (\ref{eqn5.1}), (\ref{eqn5.2}),  (\ref{eqn5.3}).

Let $\Gamma$  be  any open  subset  of  $\partial\D$.   We say that  boundary measurements are performed on $\R\times\Gamma$  if
we are  able  to measure  the restriction $\Lambda f\big|_{\R\times\Gamma}$  for any smooth input  $f$  with support  in $\R\times\overline\Gamma$.

Let $x'=\phi(x)$ be a diffeomorphism  of  $\overline \D$  onto $\overline\D$  such  that  $\phi(x)=x$  on $\Gamma$.  Let  $a(x)\in C^\infty(\overline\Omega)$ 
 be such  that $a(x)=0$  on $\Gamma$.   It is well-known  that  if  we  make  a change of variable
 \begin{equation}\label{eqn5.5}
 x'=\phi(x),\ \ \ x_0'=x_0+a(x),
 \end{equation}
then  in coordinates $(x_0',x')$  we get  an initial-boundary  value  problem  similar  to (\ref{eqn5.1}), (\ref{eqn5.2}),  (\ref{eqn5.3})
such that 
\begin{equation}\label{eqn5.6}
\Lambda f\big|_{\R\times\Gamma}=\Lambda'f\big|_{\R\times \Gamma},
\end{equation}
for all $f$
with support  in $\R\times\overline \Gamma$,   where  $\Lambda'$  is the DN  operator  in  $(x_0',x')$  coordinates.   Therefore  we have  to study  
the determination  of the metric  from  boundary  measurements  on  $\R\times\Gamma$  only  modulo  changes  of variables  of the  form  (\ref{eqn5.5}).
It was proven in [Esk10b]   for  $n\geq 2$  that 
boundary measurements  on $\R\times \Gamma$  allow  recovery  of the ergosphere   
$\partial \Omega=\{ g_{00}=0\}$   and  the metric  on $\overline \D\setminus \Omega$  up  to changes of variables (\ref{eqn5.5}).

It follows  from the considerations  in [Esk08]  that if at least  one point of $\partial\Omega$  is characteristic,  then it is necessary  to spend  
infinite time to recover
the ergosphere  $\partial\Omega$,  i.e.  for any $T$,  boundary  measurements on $[0,T]\times \Gamma$  do not determine $\partial\Omega$  in a neighborhood  of 
the characteristic points.  However,  if all points  of  $\partial\Omega$  are not characteristic then there exists $T_0$  such that boundary  measurements  on  
$[0,T_0]\times\Gamma$  determine the ergosphere and the metric $[g_{jk}]_{j,k=0}^n$  on the ergosphere  up to diffeomorphisms (\ref{eqn5.5}).

In this section,  in the case $n=2$,  we expand  upon  the result of  [Esk10b],  treating the recovery  of a black hole  inside $\Omega$.  It follows  from  the results 
of Section 3  that  if  $b_1(\theta)<0$  (see  (\ref{eqn1.2}), (\ref{eqn1.3}),  (\ref{eqn1.4}))  and if  $\partial\Omega$  is  not characteristic   then there exists a black hole  
  $\Omega_0$  inside  $\Omega$  and the black  hole event  horison   $\partial\Omega_0$  is smooth.

Note that  the  equations
for  black holes depend  only on  the spatial  part  
$G=
\begin{bmatrix}
g^{\rho\rho} & g^{\rho\theta} \\
g^{\theta\rho} & g^{\theta\theta}
\end{bmatrix}
$
of the inverse metric tensor  $[g^{ij}]_{i,j=0}^2$.   Introduce  coordinates  $(\rho,\theta)$  in $\overline \Omega\setminus \Omega_0,\ \theta\in \R/2\pi\mathbb Z,
\ 0\leq \rho\leq \rho_0(\theta)$  in  $\overline \Omega\setminus\Omega_0$,  extending  those in Section 2.2,   so that  $\rho =-\Delta$  near  $\partial\Omega$  
and  $\rho=0$  is  the equation  of  $\partial\Omega$,  and  chosen so   that the event  horizon $\partial\Omega_0$  is  a graph  given   by  $\rho=\rho_0(\theta)$.

Consider the equation  for characteristics $\phi^\pm=\phi^\pm(\rho,\theta)$:
\begin{equation}\label{eqn5.7}
g^{\rho\rho}(\phi_\rho^\pm)^2+2g^{\rho\theta}\phi_\rho^\pm\phi_\theta^\pm+g^{\theta\theta}(\phi_\theta^\pm)^2=0.
\end{equation}
It follows  from (\ref{eqn5.7})  that the matrices  $G(\rho,\theta)$  and  $\lambda(\rho,\theta)G(\rho,\theta)$,  where   $\lambda(\rho,\theta)\neq 0$,  
produce the same characteristics 
equation,  i.e.  the characteristic  equations  (and black holes)  do not  depend on the scaling  factor  $\lambda(\rho,\theta)$.

Thus assuming that  $g^{\rho\rho}\neq 0$  for all  $0\leq \rho\leq \rho_0$,  we get 
\begin{equation}\label{eqn5.8}
\phi_\rho^\pm=\frac{-g^{\rho\theta}\pm\sqrt{\rho_1}}{g^{\rho\rho}}\phi_\theta^\pm,
\end{equation}
where  we have  used  that  $g^{\rho\rho}g^{\theta\theta}-(g^{\rho\theta})^2=-\rho_1,\ \rho_1=C^2\rho$.   We impose the following  boundary
conditions  on  $\phi^+(\rho,\theta)$  and  $\phi^-(\rho,\theta)$  when  $\rho=0$:
\begin{equation}\label{eqn5.9}
\phi^\pm(0,\theta)=\theta,\ \ \ \theta\in [0,2\pi].
\end{equation}
Consider  the  curves  $\phi^+(\rho,\theta)=\theta_0,\ \ \phi^-(\rho,\theta)=\theta_0$   for fixed  $\theta_0\in \R/2\pi\mathbb Z$.   It was shown  in 
[Esk10]  (see also Section 3)   that  it is  possible   to use  the time  variable  $x_0$   as a parameter   for both curves.  One  of the curves,  say  
$\phi^-(\rho,\theta)=\theta_0$,  starts  at  $(0,\theta_0)$   when $x_0=t_0$  and approaches  the singularity  at  $0\in  \Omega_0$   when 
$x_0\rightarrow +\infty$,  crossing the  event   horizon  $\partial \Omega_0$   at some time $t$.  The second  curve  $\phi^+(\rho,\theta)=\theta_0$  
ends  at  $(0,\theta_0)$  as  $x_0$  increases.  When  $x_0\rightarrow  -\infty$,  the  curve  $\phi^+(\rho,\theta)=\theta_0$   spirals  around  the event
horizon   $\partial\Omega$.  For  definiteness supppose  $\phi^+=\theta_0$  spirals  
counter-clockwise  when  $x_0\rightarrow -\infty$.  Note that 
\begin{equation}\label{eqn5.10}
g^{\rho\rho}\phi_\rho^\pm+g^{\rho\theta}\phi_\theta^\pm=0\ \ \mbox{at}\ \ (0,\theta_0),
\end{equation}
Note that $\rho(\theta)$  is a periodic function on  $(-\infty,\infty)$.  Make  the change of variables
\begin{equation}\label{eqn5.11}
 \sigma=\phi^+(\rho,\theta),\ \ \ \ \tau=\phi^-(\rho,\theta),
 \end{equation}
 where  $(\rho,\theta)\in \Pi$.  The Jacobian  is
 \begin{equation}\label{eqn5.12}
 \frac{\partial(\sigma,\tau)}{\partial(\rho,\theta)}=\phi_\rho^+\phi_\theta^- -\phi_\rho^-\phi_\theta^+=
 \frac{(-g^{\rho\theta}+\sqrt{\rho_1})\phi_\theta^+\phi_\theta^-}{g^{\rho\rho}}-
 \frac{(-g^{\rho\theta}-\sqrt{\rho_1})\phi_\theta^-\phi_\theta^+}{g^{\rho\rho}}=
  \frac{2\sqrt{\rho_1}}{g^{\rho\rho}}\phi_\theta^+\phi_\theta^-.
\end{equation}
Therefore   the map  (\ref{eqn5.11})   is one-to-one  when  $\rho>0$  and it  is not  smooth  when  $\rho=0$  and  it is  not smooth when  $\rho=0$.
Note  that $\phi^+(\rho,\theta)=\theta_0$  approaches  $+\infty$  when  $x_0\rightarrow -\infty$.   Make a new change  of variables
\begin{equation}\label{eqn5.13}
y_1=\frac{\sigma+\tau}{2}=\frac{\phi^+(\rho,\theta)+\phi^-(\rho,\theta)}{2},
\ \ \ y_2=\frac{\sigma-\tau}{2}=\frac{\phi^+(\rho,\theta)-\phi^-(\rho,\theta)}{2}.
\end{equation}
Note  that 
\begin{equation}\label{eqn5.14}
y_1\big|_{\rho=0}=\theta,\ \ y_2\big|_{\rho=0}=0,\ \ \ \ -\infty<\theta<\infty.
\end{equation}
We have that  $y_2\rightarrow +\infty$  when  $\sigma\rightarrow +\infty$.  Denote  the map (\ref{eqn5.13}) by $\Phi$.   Thus  $\Phi$  maps  $\Pi$  
onto  the half-plane  $\R^2=\{-\infty<y_1<\infty,\ \ y_2>0\}$.   It follows   from  (\ref{eqn5.14})   that the map $\Phi$  is the identity
on  $\{\rho=0\}$.   Characteristic  curves $\phi^+=c_+,\ \phi^-=c_-$  become  $y_1+y_2=c_+,\ y_1-y_2=c_-$  after the map  $\Phi$.
Varying  $c_+,c_-$  we can  fill  the half-plane  $\{y_1\in \R,y_2\geq 0\}$.
Note  that the event  horizon  $\rho=\rho_0(\theta)$  is the boundary  of the strip  $\Pi$.
Note  that  the matrix  $G$  has  the following  form  after  applying  the map $\Phi$:
\begin{equation}\label{eqn5.15}
G'=\frac{1}{4}\hat g^{\sigma\tau}
\begin{bmatrix}
-1 & 0 \\
0 & 1
\end{bmatrix}
=\Phi G\Phi^t,
\end{equation}
where 
\begin{equation}\label{eqn5.16}
\begin{aligned}
\hat g^{\sigma\tau}
&=g^{\rho\rho}\phi_\rho^+\phi_\rho^- +g^{\rho\theta}(\phi_\theta^+\phi_\rho^- +\phi_\rho^+\phi_\theta^-)+g^{\theta\theta}\phi_\theta^+\phi_\theta^-
\\
&=\Bigg[g^{\rho\rho}
\frac{(-g^{\rho\theta}+\sqrt{\rho_1})}{g^{\rho\rho}}
\ \frac{(-g^{\rho\theta}-\sqrt{\rho_1})}{g^{\rho\theta}}
\\
&+g^{\rho\theta}
\Bigg(\frac{-g^{\rho\theta}+\sqrt{\rho_1}}{g^{\rho\rho}}
+
\frac{-g^{\rho\theta}-\sqrt{\rho_1}}{g^{\rho\rho}}\Bigg)
+g^{\theta\theta}\Bigg]\phi_\theta^+\phi_\theta^-
\\
&=\Bigg[
\frac
{(g^{\rho\theta})^2-\rho_1}
{g^{\rho\rho}}
+\frac{-2(g^{\rho\theta})^2}{g^{\rho\rho}}+g^{\theta\theta}\Bigg]\phi_\theta^+\phi_\theta^-
\\
&=
\frac{-2C^2\rho}{g^{\rho\rho}}\phi_\theta^+\phi_\theta^-,
\end{aligned}
\end{equation}
since  $\rho_1=C^2\rho$.

Suppose  we have
another metric  $[g_1^{jk}(\rho,\theta)]_{j,k=0}^2$ having  the same  
boundary measurements  on $\R\times\Gamma$.  Then
 metrics $g,g_1$  are  the same  in $\D\setminus \Omega$  up to  changes  of variables (\ref{eqn5.5})  (cf. [Esk10b]).
Therefore  we may  assume  that  the ergosphere $\partial\Omega$  for both metrics is the same  and the restriction of both metrics to
$\partial\Omega$  is also  the same.  Suppose  $\phi_1^\pm$  satisfy  (\ref{eqn5.8}),  (\ref{eqn5.9})  with $g$  replaced  by  $g_1$.
Let  $\phi_1^\pm(0,\theta)=\theta$  for all $\theta\in \R/2\pi\mathbb Z$.  Make the change  of variables 
$\sigma_1=\phi_1^+(\rho,\theta),\ \tau_1=\phi_1^-(\rho,\theta),\ \theta\in \R,\ 0\leq \rho_0^{(1)}(\theta)$,  where $\rho=\rho_0^{(1)}(\theta)$  is the equation  
of the event horizon $\partial\Omega_0'$.  Make  also  the change of variables
\begin{equation}\label{eqn5.17}
y_1'=\frac{\sigma_1+\tau_1}{2},\ \ \ y_2'=\frac{\sigma_1-\tau_1}{2}.
\end{equation}
Denote the map (\ref{eqn5.17})  by $\Phi_1$.  Thus $\Phi_1$  maps  $\Pi'=\{\theta\in\R,\ 0\leq \rho <\rho_0'(\theta)\}$  onto  $\R_+^2=\{y_1\in \R,\ y_2'\geq 0\}$. 
 Note that  $\Phi_1$  is a homeomorphism 
and $\Phi_1$  is  the identity  on  $\{-\infty<\theta<\infty,\ \rho=0\}$.

Let $G_1=\begin{bmatrix}
g_1^{\rho\rho} & g_1^{\rho\theta} \\
g_1^{\rho \theta}& g_1^{\theta\theta}
\end{bmatrix}
$
be the spacial  part  of the inverse  metric tensor  $g_1$.   Making the change  of variables (\ref{eqn5.17})
we get  analogously  to (\ref{eqn5.15})
\begin{equation}\label{eqn5.18}
G_1'=\frac{1}{4}\hat g_1^{\sigma_1\tau_1}
\begin{bmatrix}
-1 & 0 \\
0 & 1
\end{bmatrix}
=\Phi_1 G_1\Phi_1^t,
\end{equation}
where  analogously to (\ref{eqn5.16})
\begin{equation}\label{eqn5.19}
\hat g_1^{\sigma_1\tau_1}=
\frac{-2C_1^2\rho}{g_1^{\rho\rho}}\phi_{1\theta}^+\phi_{1\theta}^-,\ \ \  C_1>0.
\end{equation}
Combining  (\ref{eqn5.15})   and (\ref{eqn5.18})  we get
\begin{equation}\label{eqn5.20}
G_1=\lambda\Phi_1^{-1}\Phi G(\Phi_1^{-1}\Phi)^t
\end{equation}
where  
\begin{equation}\label{eqn5.21}
\lambda=\hat g_1^{\sigma_1\tau}(\hat g^{\sigma\tau})^{-1}.
\end{equation}
It follows  from   
 (\ref{eqn5.16}),  (\ref{eqn5.19})
that  $\lambda\neq 0$  for  $\rho\geq  0$  and  smooth.

Analogously  to (\ref{eqn5.1})
we have that  $\frac{\partial(\sigma_1,\tau_1)}{\partial (\rho,\theta)}=\frac{2C_1\sqrt \rho}{g_1^{\rho\rho}}\phi_{1\theta}^+\phi_{1\theta}^-$.   Thus
$\frac{\partial(\sigma_1,\tau_1)}{\partial (\sigma,\tau)}
=\frac{\partial(\sigma_1,\tau_1)}{\partial (\rho,\theta)}\big(\frac{\partial(\sigma,\tau)}{\partial (\rho,\theta)}\big)^{-1}\neq 0$
and smooth for  $\rho\geq  0$.

Therefore  the map  $\Phi_1^{-1}\Phi$  of  $\Pi=\{\theta\in \R,\ 0\leq \rho<\rho_0(\theta)\}$  into  $\Pi'=\{\theta\in \R,\ 0\leq \rho<\rho_0'(\theta)\}$  is
a diffeomorphism  and  $\Phi_1^{-1}\Phi$  is  the identity  on  $\{\rho=0,\theta\in \R\}$.

Taking the closure  of  $\Pi$  and  $\Pi'$ we  get  a diffeomorphism  of the event  horizons  $\partial\Omega_0$  and $\partial\Omega_0'$.
Thus  we have proven  that the event  horizon  $\partial\Omega_0$  is determined  uniquely  up  to diffeomorphism equal  to the identity  on the ergosphere.

\newcommand{\etalchar}[1]{$^{#1}$}


\begin{thebibliography}{}

\bibitem[BCO{\etalchar{+}}11]{Belgiorno2011dielectric}
F~Belgiorno, SL~Cacciatori, G~Ortenzi, L~Rizzi, V~Gorini, and D~Faccio.
\newblock Dielectric black holes induced by a refractive index perturbation and
  the hawking effect.
\newblock {\em Physical Review D}, 83(2):024015, 2011.

\bibitem[BLV{\etalchar{+}}05]{BarceloLiberatiVisser}
Carlos Barcel{\'o}, Stefano Liberati, Matt Visser, et~al.
\newblock Analogue gravity.
\newblock {\em Living Rev. Rel}, 8(12):214, 2005.

\bibitem[Esk08]{EskinOptical}
Gregory Eskin.
\newblock Optical Aharonov-Bohm effect: an inverse  hyperbolic  problems approach.
\newblock {\em Comm. Math. Phys.}, 284(2):317-343--839, 2008.


\bibitem[Esk10]{EskinInvHyp}
Gregory Eskin.
\newblock Inverse hyperbolic problems and optical black holes.
\newblock {\em Comm. Math. Phys.}, 297(3):817--839, 2010.


\bibitem[Esk10b]{EskinInvKerr}
Gregory Eskin.
\newblock Perturbations  of the  Kerr black  hole  and  boundness of linear waves.
\newblock {\em J. Math. Phys.}, 57(11):112501, 2010.


\bibitem[Esk14]{EskinNonstationary}
Gregory Eskin.
\newblock Nonstationary artificial black holes.
\newblock {\em Inverse Problems}, 30, 125007, 2014.

\bibitem[FFL{\etalchar{+}}10]{FFLKT}
Serena Fagnocchi, Stefano Finazzi, Stefano Liberati, Marton Kormos, and Andrea
  Trombettoni.
\newblock Relativistic bose-einstein condensates: a new system for analogue
  models of gravity, 2010.

\bibitem[FN98]{FrolovNovikov1998}
Valeri Frolov and Igor Novikov.
\newblock {\em Black hole physics: basic concepts and new developments},
  volume~96.
\newblock Springer, 1998.

\bibitem[Gor23]{Gordon}
Walter Gordon.
\newblock Zur lichtfortpflanzung nach der relativitätstheorie.
\newblock {\em Annalen der Physik}, 377(22):421--456, 1923.

\bibitem[Hal13]{HallThesis}
Michael~A Hall.
\newblock Phd thesis.
\newblock UCLA, 2013.

\bibitem[LP99]{LP}
Ulf Leonhardt and Paul Piwnicki.
\newblock Relativistic effects of light in moving media with extremely low
  group velocity.
\newblock {\em arXiv preprint cond-mat/9906332}, 1999.

\bibitem[NVV02]{NovelloVisserVolovik}
M{\'a}rio Novello, Matt Visser, and Grigory~E Volovik.
\newblock {\em Artificial black holes}.
\newblock World Scientific Publishing Company, 2002.

\bibitem[PKR{\etalchar{+}}08]{Philbin2008}
Thomas~G Philbin, Chris Kuklewicz, Scott Robertson, Stephen Hill, Friedrich
  K{\"o}nig, and Ulf Leonhardt.
\newblock Fiber-optical analog of the event horizon.
\newblock {\em Science}, 319(5868):1367--1370, 2008.

\bibitem[RMM{\etalchar{+}}10]{RousseauxLeonhardt2010}
Germain Rousseaux, Philippe Ma{\"\i}ssa, Christian Mathis, Pierre Coullet,
  Thomas~G Philbin, and Ulf Leonhardt.
\newblock Horizon effects with surface waves on moving water.
\newblock {\em New Journal of Physics}, 12(9):095018, 2010.

\bibitem[SU02]{SchutzholdUnruh2002}
Ralf Sch{\"u}tzhold and William~G Unruh.
\newblock Gravity wave analogues of black holes.
\newblock {\em Physical Review D}, 66(4):044019, 2002.

\bibitem[Unr81]{Unruh}
William~George Unruh.
\newblock Experimental black-hole evaporation?
\newblock {\em Phys. Rev. Lett.}, 46:1351--1353, May 1981.

\bibitem[Vis98]{VisserAcoustic}
Matt Visser.
\newblock Acoustic black holes: horizons, ergospheres and {H}awking radiation.
\newblock {\em Classical Quantum Gravity}, 15(6):1767--1791, 1998.

\bibitem[Vis12]{Visser2012}
Matt Visser.
\newblock Survey of analogue spacetimes.
\newblock Lecture Notes in Physics Volume 870 (2013) 31-50, 2012.

\bibitem[VMP10]{VisserMolinaParis2010}
Matt Visser and Carmen Molina-Paris.
\newblock Acoustic geometry for general relativistic barotropic irrotational
  fluid flow, 2010.

\bibitem[Wal10]{Wald2010}
Robert~M Wald.
\newblock {\em General relativity}.
\newblock University of Chicago press, 2010.

\end{thebibliography}


\end{document}